\documentclass[journal]{IEEEtranTCOM}
\normalsize
\usepackage{times}

\ifCLASSINFOpdf
\else
\fi

\usepackage{multicol}
\usepackage{bbold}
\usepackage[tight,footnotesize]{subfigure}
\usepackage{subfigmat}
\usepackage[english]{babel}
\usepackage[T1]{fontenc}
\usepackage[utf8]{inputenc}
\usepackage{multirow}
\usepackage{xcolor}
\usepackage{amssymb}
\usepackage[cmex10]{amsmath}
\usepackage{graphicx}
\makeatletter
\newcommand*\bigcdot{\mathpalette\bigcdot@{1}}
\newcommand*\bigcdot@[2]{\mathbin{\vcenter{\hbox{\scalebox{#2}{$\m@th#1\bullet$}}}}}
\makeatother
\usepackage{mathtools}
\usepackage{cite} 
 
\usepackage{upref}
\usepackage{mathtools}

\DeclarePairedDelimiter\floor{\lfloor}{\rfloor}
\DeclarePairedDelimiter\abs{\lvert}{\rvert}%
\usepackage{varwidth}
\usepackage{subfigure}
\usepackage{array,booktabs} 
\newcolumntype{P}[1]{>{\centering\arraybackslash}p{#1}}
\newcolumntype{M}[1]{>{\centering\arraybackslash}m{#1}}
\usepackage[linesnumbered,ruled,vlined]{algorithm2e}
\makeatletter
\newcommand{\nosemic}{\renewcommand{\@endalgocfline}{\relax}}
\newcommand{\dosemic}{\renewcommand{\@endalgocfline}{\algocf@endline}}
\let\oldnl\nl
\newcommand{\nonl}{\renewcommand{\nl}{\let\nl\oldnl}}
\makeatother

\usepackage{MnSymbol}

\usepackage{amsfonts}

\usepackage{mathalfa}
\DeclareMathAlphabet{\mathcalligra}{T1}{calligra}{m}{n}

\usepackage{amsthm}
\newtheorem{proposition}{Proposition}
\newtheorem{theorem}{Theorem}

\usepackage{filecontents}
\usepackage{caption}

\usepackage{float}

\usepackage{algpseudocode}

\DeclareMathOperator*{\argmax}{arg\,max}

\usepackage{calligra}
\usepackage[T1]{fontenc}
\usepackage{mathrsfs}
\usepackage{epstopdf}
\graphicspath{{figures/}}

\begin{document}
\title{Throughput-based Design for Polar Coded-Modulation}

\author{Hossein~Khoshnevis,~\IEEEmembership{Student Member,~IEEE,}
             Ian~Marsland,~\IEEEmembership{Member,~IEEE,}
            and~Halim~Yanikomeroglu,~\IEEEmembership{Fellow,~IEEE}
\thanks{H. Khoshnevis, I. Marsland, and H. Yanikomeroglu are with the Department
of Systems and Computer Engineering, Carleton University, Ottawa, ON, Canada
(e-mail: \{khoshnevis, ianm, halim\}@sce.carleton.ca). This work is supported by Huawei Canada Co., Ltd.  This paper was presented in part in \textit{IEEE VTC-Fall} 2017.}}


\maketitle
\begin{abstract}
Typically, forward error correction (FEC) codes are designed based on the minimization of the error rate for a given code rate. However, for applications that incorporate hybrid automatic repeat request (HARQ) protocol and adaptive modulation and coding, the throughput is a more important performance metric than the error rate. Polar codes, a new class of  FEC codes with simple rate matching, can be optimized efficiently for maximization of the throughput. In this paper, we aim to design HARQ schemes using multilevel polar coded-modulation (MLPCM). Thus, we first develop a method to determine a set-partitioning based bit-to-symbol mapping for high order QAM constellations. We simplify the LLR estimation of set-partitioned QAM constellations for a multistage decoder, and we introduce a set of algorithms to design throughput-maximizing MLPCM for the successive cancellation decoding (SCD). These codes are specifically useful for non-combining (NC) and Chase-combining (CC) HARQ protocols. Furthermore, since optimized codes for SCD are not optimal for SC list decoders (SCLD), we propose a rate matching algorithm to find the best rate for SCLD while using the polar codes optimized for SCD. The resulting codes provide throughput close to the capacity with low decoding complexity when used with NC or CC HARQ.
\end{abstract}

\IEEEpeerreviewmaketitle

\section{Introduction}
\label{sec:Introduction}

\IEEEPARstart{T}{he} time-varying nature of wireless channels requires the use of adaptive modulation and coding (AMC) schemes to achieve high throughput communication. Hybrid automatic repeat request (HARQ) error control protocols can enhance the throughput of AMC especially when high order modulation schemes are used. Typically, the design objective of binary error correction codes for the additive white Gaussian noise (AWGN) channel is to minimize the error rate, for a given code rate and signal-to-noise ratio (SNR). This method of design has been widely used for convolutional codes \cite{Lin1983}, parallel concatenated (turbo) codes \cite{Divsalar1995}, low density parity check codes (LDPC) \cite{MacKay1996} and polar codes \cite{Arikan2009,Tal2013}.  However, to optimize the AMC and HARQ protocols, the throughput is a much more relevant performance metric than error rate, where throughput is defined as the average rate of successful message delivery and indicates how close the performance of a system is to the channel capacity. Designing throughput-optimal codes and coded-modulation typically involves an exhaustive search over a set of code rates and modulations and employs simulation to estimate the throughput \cite{Deng1995}. However, as demonstrated in this paper, this process can be greatly simplified for polar codes. 

Polar codes, introduced by Erdal Arikan in \cite{Arikan2009}, are a new class of forward error correction (FEC) codes that can achieve the capacity of binary-input symmetric-output memoryless channels under successive cancellation decoding (SCD) as their length tends to infinity. However, finite length polar codes can be designed and decoded efficiently to perform close to the channel capacity. Polar codes work based on the concept of channel polarization which simplifies the rate matching. Therefore, they are a strong candidate for designing throughput-maximizing FEC codes. When designing polar codes, a set of information bit-channels must be chosen. To find this information-set, the distribution of the log-likelihood ratios (LLRs) of bit-channels should be determined using simulation \cite{Arikan2009} or density evolution. Trifonov in  \cite{Trifonov2012} showed that the density evolution based design of polar codes in an AWGN channel can be simplified using the Gaussian approximation (GA) of the distribution of  LLRs.

When designing polar codes based on the throughput, the optimal set of information bits to maximize the throughput is chosen. The particular advantage of polar codes that facilitates their optimization for maximizing the throughput is this straightforward design method in comparison to most other modern codes. These polar codes are particularly useful for non-combining (NC)\footnote[1]{Non-combining HARQ is also known as Type-I HARQ.} and Chase-combining (CC) HARQ where, for the retransmission of a failed codeword, the whole codeword should be retransmitted. Furthermore, they can be used for incremental redundancy HARQ (IR HARQ) schemes. Although efficient IR HARQ schemes using polar codes have been proposed (e.g.,\cite{Saber2015}) their scheduling, especially for coded-modulation schemes, is difficult \cite{Liang2017} since transmitting short length redundancy may need an overwhelming amount of control signals. To achieve simple scheduling, we set the retransmission redundancy sizes to be equal to the mother code length. In this case, limited feedback can be employed to achieve high throughput using AMC and HARQ protocols since the polar code graph is fixed and only the information-set should be modified when the SNR changes. Polar codes have also been designed to maximize the throughput of CC HARQ with SCD in \cite{Chen2014} and with SCLD in \cite{Liang2017} based on puncturing for a fixed message length and a BPSK constellation. These methods, given the message length $K$, require a full search over all puncturing lengths from $0$ to $N-K$ and thus require running the GA method for $N-K$ times. However, in our proposed code design method, we only need to run the GA method once for most protocols. Therefore, their design is at least $N-K$ times more complex than our method and, as we will show, there is no advantage in using their algorithms. 

To construct polar coded-modulation, either multilevel coding with multistage decoding (MLC/ MSD) or bit-interleaved-coded-modulation (BICM), two well-known methods of concatenating binary codes with high order modulations, can be employed, but MLC/MSD with a set-partitioning based bit-to-symbol mapping (SPM), due to the clarity of design and the conceptual similarity to channel polarization, outperforms BICM \cite{Seidl2013-1}. In addition, the LLR estimation for MSD is cheaper than for BICM which make it a promising technique for future wireless communication systems. Similar to BPSK modulation, designing multilevel polar coded-modulation (MLPCM) using the puncturing-based search methods in \cite{Liang2017} and \cite{Chen2014}  is hard since they require a full search over all levels of a MLC. To design MLPCM, Trifonov  in \cite{Trifonov2012} suggested using the GA, but the full steps of design have not been presented in the literature. 

In this paper, first the accurate LLR estimation of quadrature amplitude modulation (QAM) with two dimensional (2D) SPM is simplified using a novel device constructed from two independent  pulse amplitude modulation (PAM) constellations for the I-channel and Q-channel, and a linear bit mapping. Then, an approximation is proposed to simplify the LLR estimation of PAM constellations. Using the LLR simplifications, the average LLR of each binary channel of QAM constellations is measured and a set of algorithms for designing MLPCM using GA is proposed that is based on maximizing the throughput, instead of the well-studied objective of minimizing the FER. In these algorithms, we fix the code-length and find the optimal information-set using bounds on the throughput of the coded-modulation schemes. These codes are designed for NC, CC, and IR HARQ schemes with successive cancellation decoders (SCD).  In addition, since polar codes optimized for SCD are suboptimal for SC list decoding (SCLD), introduced in \cite{Tal2015}, a fast rate matching algorithm is proposed to find the code rate corresponding to the maximum throughput for SCLD when used with polar codes designed for SCD.

In particular, first we show that by adapting the MLPCM to SNR, throughput very close to the capacity can be achieved. In this case, CC HARQ does not provide any advantage over NC HARQ. However, when the transmitter is restricted to use a smaller number of codes, CC HARQ outperforms NC HARQ and also provides higher throughput in comparison to BICM-based CC and IR HARQ, constructed in \cite{Elkhami2015}. We also show that when the levels of MLPCM can be decoded independently using HARQ, the throughput is enhanced substantially.

The rest of the paper is organized as follows: The system model and HARQ protocols are described in Section~\ref{sec:sysmodel}, and the simplified LLR approximation method for QAM and PAM constellations with SPM along with a simple method of constructing set-partitioning for QAM constellations is described in Section~\ref{LLRSimple}. The throughput of HARQ protocols as a design metric is explained in Section~\ref{sec:Throughput}. The polar code design methods based on the throughput for a BPSK constellation with SCD are introduced in Section~\ref{polarcodedesign}, and the MLPCM design procedure for QAM constellations is described in Section~\ref{polarcodedesignQAM}. In addition, a rate matching algorithm for SCLD is proposed in Section~\ref{polarcodedesignSCL}. Finally, numerical results and discussions are provided in Section~\ref{sec:num}, and conclusions are presented in Section~\ref{sec:Conclusion}.

\section{System Model}
\label{sec:sysmodel}

The communication system includes a single user transmitter and receiver that use a NC or CC HARQ error control protocol. At the transmitter, a cyclic redundancy check (CRC) sequence of length $L_{\text{CRC}}$ is added to data of length $K$ bits and each $K^{\prime}=K+L_{\text{CRC}}$ bits of the CRC and data are coded using a multilevel polar code of total length $N_{\text{tot}}$ and a total code rate of $R=K^{\prime}/N_{\text{tot}}$. The MLC scheme employs a set of $B$ independent encoders, each with a code length $N=N_{\text{tot}}/B$ and code rate $R_n$, with one encoder for each binary channel corresponding to an address bit of a constellation with a cardinality of $2^B$. After encoding of all levels, each set of code-bits \{$c_1^i,c_2^i,...,c_B^i$\}, for $i=1,...,N$, are modulated using BPSK or QAM and transmitted through an AWGN channel with a noise variance of $\frac{N_0}{2}$ per dimension. Throughout this paper, due to the power inefficiency of rectangular  QAM when $B$ is an odd number, we limit the discussions to even values of $B$. The system model can be written as 
\begin{equation}
\label{sysmodel}
 y_{i,l}= x_{i,l}+w_{i,l},
\end{equation}
where $y_{i,l}$ is the $i^{th}$ received sample in the $l^{th}$ retransmission and $x_{i,l}$ and $w_{i,l}$ represent the transmitted symbol and the noise, respectively. The system has an average SNR of $\gamma=E_s/N_0$ where $E_s=E[|x_{i,l}|^2]$.

In a MSD, the LLR for each level is calculated given the code-bits for the upper-levels \cite{Imai1977}.  To this end, after decoding of each level, the received message word is re-encoded and is fed to the demapper to reduce the LLR estimation ambiguity of the next level since the knowledge about the upper-level bits limits the cardinality of lower-level symbol-sets for the LLR calculation. The LLR calculation at each level can be written as 
\begin{equation}
\label{LLRMSD}
\lambda_{n,i,l}= \ln \frac{ \sum_{x \in \mathcal{X}_{n}^0 } p(y_{i,l} \mid x_{i,l}=x)} {\sum_{x \in \mathcal{X}_{n}^1}  p(y_{i,l} \mid x_{i,l}=x) },
\end{equation}
where $n=1,2,...,B$, $\mathcal{X}_{n}^0$ and $\mathcal{X}_{n}^1$  are the set of all constellation points with zero or one in their $n^{th}$ position given the upper-level code-bits, respectively, and $p(y_{i,l} \mid x_{i,l}=x)$ is the likelihood function of $y_{i,l}$ being received given $x$ is transmitted.

\subsection{HARQ Protocols}
\label{sec:protocols}

In this paper, we design MLPCM for four HARQ protocols. The protocols are divided into two groups,  level-dependent and level-independent, based on the dependency of levels for decoding. 

\noindent \textbf{Level-Dependent HARQ Protocols}

In level-dependent protocols, all levels of a multilevel codeword are decoded and an  ACK is fed back to the transmitter only if all levels are correct. Indeed, all levels of the multilevel code are dependent and all levels are seen as one codeword. Thus, only one CRC for checking the correctness of the codeword is employed.

The LLRs used for each decoding attempt depend on whether NC or CC is used. For NC level-dependent (NC-D) HARQ, the LLRs depend only on the received samples for the latest retransmission of the codeword, so the LLRs are given by (\ref{LLRMSD}). For CC level-dependent (CC-D) HARQ, the LLRs depend on the received samples from all retransmissions, according to 
\begin{equation}
\label{LLRMSDCC}
\lambda_{n,i,L}= \ln \frac{ \sum_{x \in \mathcal{X}_{n}^0} \prod_{l=1}^{L} p(y_{i,l} \mid x_{i,l}=x)} {\sum_{x \in \mathcal{X}_{n}^1} \prod_{l=1}^{L} p(y_{i,l} \mid x_{i,l}=x) },
\end{equation}
where $L$ is the number of retransmissions.

\noindent \textbf{Level-Independent HARQ Protocols}

In level-independent protocols, as proposed in \cite{Luo2003}, independent codewords with their own CRC are transmitted on each level of MLC. For each new codeword, the receiver decodes the codeword of each level independently and checks whether it is a valid codeword. When a codeword of a specific level is invalid, the same codeword is retransmitted during the next transmission on the same level while on  all other upper and lower levels, new codewords containing new messages are transmitted. In this protocol, the receiver waits until a codeword of a level is decoded correctly before attempting to decode the next levels.

An example of the protocol is shown in Fig.~\ref{figp1:b}. In this example, codewords $A_1$, and $B_1$, are transmitted using the first and second levels of a two-level MLC, and codeword $A_1$ is decoded incorrectly. No attempt is made to decode $B_1$, (since its LLRs cannot be calculated without reliable knowledge of $A_1$), so a NACK is sent to the transmitter informing that $A_1$ failed. The transmitter responds by retransmitting $A_1$ on level 1, and transmitting a new level-2 codeword, $B_2$, on level 2. In this example the receiver successfully decodes $A_1$ after this transmission, so now it can attempt to decode $B_1$ (using the received samples from the previous transmission) and $B_2$ (using the current received samples). If $B_1$ fails, the transmitter is instructed to send $A_2$ and $B_1$. If $B_2$ fails, the transmitter sends $A_2$ and $B_2$, and if $B_1$ and $B_2$ succeed, $A_2$ and $B_3$ are transmitted (as shown in Fig.~\ref{figp1:b}).

The throughput of level-independent protocols is expected to be more than the level-dependent protocols, because in level-dependent protocols all levels with the same message words are retransmitted while in level-independent protocols, only erroneous upper-levels are retransmitted and lower-levels are used to transmit new codewords. 

For NC level-independent (NC-I) HARQ the LLRs are calculated according to (\ref{LLRMSD}), the same as for NC-D. For CC level-independent (CC-I) HARQ, the combining scheme is a little different than that of CC-D, because different codewords may be transmitted on each level during each transmission, so the LLRs are given by
\begin{equation}
\label{LLRMSDcc3}
\lambda_{n,i,l}= \ln \frac{ \sum_{x \in \mathcal{X}_{n}^0} p(y_{i,l} \mid x_{i,l}=x)} {\sum_{x \in \mathcal{X}_{n}^1} p(y_{i,l} \mid x_{i,l}=x) }+\lambda_{n,i,l-1}.
\end{equation}

\begin{figure}
\centering   
    \includegraphics[width=0.45\textwidth]{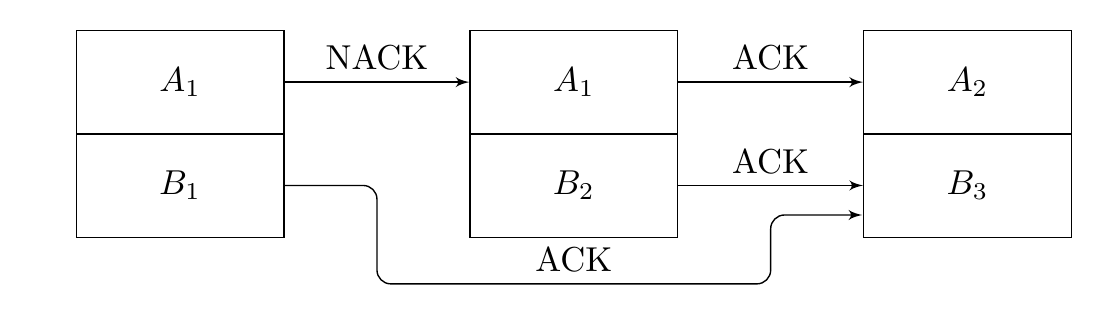}
\caption{Illustrative example of the use of level-independent HARQ.}
\label{figp1:b}
\end{figure}

For all protocols, the system uses AMC employing modulations and codes of different spectral efficiencies and rates for each SNR. The AMC is applied through the channel state information based mode-selection at the transmitter to enhance the throughput.

\section{Low Complexity LLR Estimation for MLC/MSD}
\label{LLRSimple}

Typically, at the decoder of a coded-modulation scheme, calculation of code bit LLRs using (\ref{LLRMSD}) is much more complex than actually decoding the binary FEC codes. Thus, reducing the LLR calculation complexity can substantially simplify the system. In this section, we simplify the LLR calculation for QAM and PAM constellations with SPM. The results of this section are used throughout the paper for both decoding and the code construction. For notational simplicity we drop the subscript $i$ and $l$ in the following, as we are only discussing the detection of a single symbol corresponding to one channel-use.

\subsection{LLR Calculation for Set-partitioned QAM}
\label{LLRSimpleQAM}

The LLR calculation for QAM with a two dimensional (2D) SPM needs a relationship between the LLR value and the real and imaginary parts of the received sample. This results in a function with two input variables and high arithmetic complexity. While this function can be divided into several regions and piecewise approximations with planes can be employed to simplify the LLR calculation, the number of regions grows fast for large size QAM constellations. As an alternative, a 2D SPM can be decomposed into two independent 1D SPMs. While decomposition of 2D to 1D mappings is well-known and straightforward for Gray mapping, in the following we proposed a new technique that is suitable for SPM.

We want to map $B$ bits, $[c_1c_2...c_B]$ onto a $2^B$-point square QAM constellation with SPM, such that bit $c_1$ has the lowest reliability, followed by $c_2$, and so on up to $c_B$ which is the most reliable. This can be achieved using the simple device shown in Fig.~\ref{fig:trans}, that involves a simple code and two PAM symbol mappers with natural mapping. By using this device at the transmitter, the LLR calculator at the receiver is greatly simplified. Using this device, we first precode the bits, giving $b=[b_1,b_2,...,b_B]$ where
\begin{equation}
\label{GAestimate}
b_k=\begin{cases} c_k\oplus c_{k+1} &  \forall \, \text{odd} \, k, \\ c_k &  \forall \, \text{even} \, k.\end{cases}
\end{equation} 
Then, we send the bits in odd-indexed positions to a  $2^{B/2}$-PAM symbol mapper for transmission over the I-channel, and the other bits to another $2^{B/2}$-PAM symbol mapper for the Q-channel. From Fig.~\ref{fig:trans}, $b_k^I=b_{2k-1}=c_{2k-1}\oplus c_{2k}$ and $b_k^Q=b_{2k}=c_{2k}$. 

\begin{figure}[h!]
\center
    \includegraphics[width=0.35\textwidth]{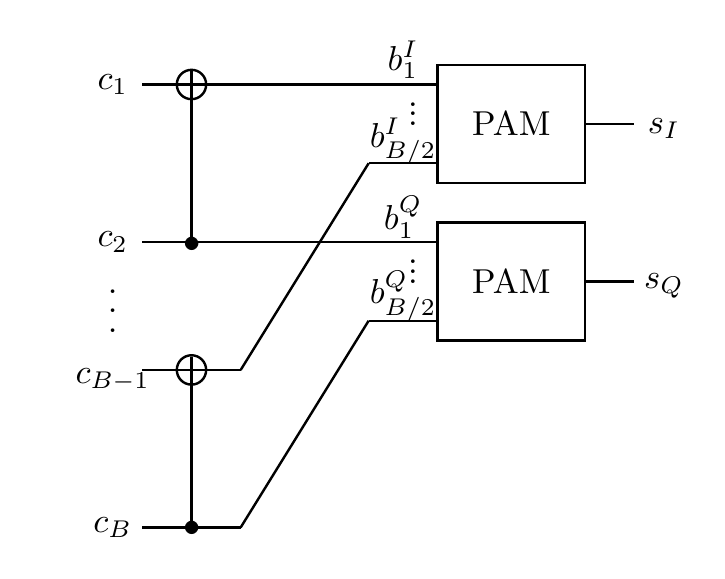}
\caption{A transform to map natural numbers to 16-QAM with 2D SPM using two independent 4-PAMs with SPM. \hfill \vspace{0.3cm}}
\label{fig:trans}
\end{figure}

The PAM symbol mappers are identical and employ natural mapping between the input bits $\textbf{d}=[d_{B/2}...d_2d_1]$ (where $d_k=b_k^I$ or $b_k^Q$ for the I-channel and Q-channel mappers, respectively), and the points 
\begin{equation}
\label{PAMLLRn}
\text{SM}_{\text{PAM}}[\textbf{d}]= 2d-(2^{B/2}-1),
\end{equation} 
where $d=d_{B/2}2^{B/2-1}+...+d_22^1+d_12^0$ is the integer representation of $\textbf{d}$. Equivalently, we can write
\begin{equation}
\label{PAMLLRn1}
\text{SM}_{\text{PAM}}[\textbf{d}]= \sum_{k=1}^{B/2} (2{d}_k-1)2^{k-1}.
\end{equation} 

Note that with this mapping, the least significant bit, $d_1$, is also the least reliable, and the most significant bit, $d_{B/2}$, is the most reliable. The transmitted QAM symbol is then
\begin{equation}
\label{QAMLLRn1}
\begin{split}
& \text{SM}_{\text{QAM}}[\textbf{c}]=\text{SM}_{\text{PAM}}[\textbf{b}^I]+j\text{SM}_{\text{PAM}}[\textbf{b}^Q] \\ & =\sum_{k=1}^{B/2}(2b_k^I-1)2^{k-1}+j\sum_{k=1}^{B/2}(2b_k^Q-1)2^{k-1} \\&= \sum_{k=1}^{B/2}(2[c_{2k-1}\oplus c_{2k}]-1)2^{k-1}+ j \sum_{k=1}^{B/2}(2c_k-1)2^{k-1},
\end{split}
\end{equation} 
where $j=\sqrt{-1}.$

\begin{theorem}
\label{Theorem11}
The resulting constellation in (\ref{QAMLLRn1}) uses SPM.
\end{theorem}
\begin{proof}
 Observe that
\begin{equation}
\label{PAMLLRn1}
\begin{split}
& \text{SM}_{2^B-\text{QAM}}[\textbf{c}] = \\
& \sum_{k=1}^{B/2-1}(2[c_{2k-1}\oplus c_{2k}]-1)2^{k-1}+j\sum_{k=1}^{B/2-1}(2c_k-1)2^{k-1} \\
& +(2[c_{B-1}\oplus c_{B}]-1)2^{B/2-1}+j(2c_B-1)2^{B/2-1} \\  
&=\text{SM}_{2^{B-2}-\text{QAM}}[\textbf{c}_{1:B-2}]+\text{SM}_{4-\text{QAM}}[\textbf{c}_{B-1:B}]2^{B/2-1},
\end{split}
\end{equation} 
so that the $2^B$-QAM constellation can be constructed by enlarging a $2^{B-2}$-QAM constellation by a factor of 4. Under the assumption that the $2^{B-2}$-QAM constellation has SPM, then the enlarged constellation does as well. Each point $x=\text{SM}_{2^{B-2}-QAM}[\textbf{c}_{1:B-2}]$ in each subset of the smaller constellation is replaced by four points in the larger constellation, $x+(-1-j)2^{B/2-1}$, $x+(1+j)2^{B/2-1}$, $x+(1-j)2^{B/2-1}$, $x+(-1+j)2^{B/2-1}$, with the same bit labellings as $x$ (i.e., $\textbf{c}_{1:B-2}$), but with 00, 01, 10 or 11 appended, respectively. Since all  the points in each set-partitioned subset of the smaller constellation fall on some regular lattice, and the new points in the enlarged subset fall on the same lattice, the minimum distance between points in the subset remains the same. By adding two more layers of the set-partitioning, creating subsets with minimum distance of $2^{B/2}$ and $\sqrt{2}\times2^{B/2}$, the resulting $2^B$-QAM constellation has SPM, provided the $2^{B-2}$-QAM constellation has SPM. But since by inspection, the 4-QAM constellation given by (\ref{QAMLLRn1}) with $B=2$ has SPM, it follows by induction that the $2^B$-QAM constellation has SPM for all even values of $B$.
\end{proof}

Note that the constellation with SPM in (\ref{QAMLLRn1}) is regular since all subsets at a specific level $n$ have the same average Euclidean distance spectrum and consecutively the same capacity \cite{Huber1994}. Therefore, one binary code can be constructed for all subsets within each level.

At the receiver, the LLRs of the code bits, $\textbf{c}$, are readily computed from the received sample. The real and imaginary parts of each received sample are sent to two different $2^{B/2}$-PAM LLR calculators, which can use (\ref{LLRMSD}) to calculate the LLRs corresponding to $\textbf{b}^I$ and $\textbf{b}^Q$, respectively. The LLRs of the first PAM bit-channels are calculated, giving $\lambda_1^{I}$ and $\lambda_1^{Q}$, which are combined to give the LLR of $c_1$,
\begin{equation}
\label{PAMN}
\lambda_1^{c}=\lambda_1^I  \boxplus \lambda_1^Q,
\end{equation}
where $\boxplus$ is the boxplus operator defined as $\lambda_1\boxplus\lambda_2= 2\,\text{tanh}^{-1}\big(\text{tanh}(\frac{\lambda_1} {2})\text{tanh}(\frac{\lambda_2} {2})\big)$. Once the upper code (of which $c_1$ is a code-bit) has been decoded, and an estimate $\hat{c}_1$ of $c_1$ based on the code has been generated, the LLR of $c_2$ can be calculated as 
\begin{equation}
\label{PAMN}
\lambda_2^{c}=(1-2\hat{c}_1)\lambda_1^I+\lambda_1^Q,
\end{equation}
and the second code can be decoded, giving $\hat{c}_2$. Armed with $\hat{c}_1$ and $\hat{c}_2$, the LLR calculators can detect the second PAM bit-channels, giving $\lambda_2^I$ and $\lambda_2^Q$. These in turn are used to calculate 
\begin{equation}
\label{PAMLLRn1}
\begin{split}
& \lambda_3^{c}=\lambda_2^I  \boxplus \lambda_2^Q, \\
&\lambda_4^{c}=(1-2\hat{c}_3)\lambda_2^I+\lambda_2^Q.
\end{split}
\end{equation} 
This process is repeated until all levels have been decoded.

\subsection{LLR Calculation for Set-partitioned PAM}
\label{LLRSimplePAM}

In last section, we showed how to separate a QAM constellation with 2D SPM into two independent PAM constellations, which by itself significantly reduces the complexity of LLR calculations, but the complexity is still needlessly high. Although the LLR calculation for BPSK is easy, computed as $\lambda=-\frac{4y}{N_0}$ for the AWGN channel, for PAM constellations direct application of (\ref{LLRMSD}) is of high complexity even for moderately sized constellations. In this section, we try to simplify the LLR estimation of PAM constellations.

The LLR calculation for a MSD based on the Jacobi theta functions is proposed in \cite{Fredj2016}, in which knowledge of the values of Jacobi theta functions is required for the LLR calculation. However, saving and referring to values of these functions needs large memory capacity. Instead, the LLR of a MSD can be approximated by the dominant term since $\ln(\underset{t}{\sum}\text{e}^{-|h_t|})\approx -\underset{t}{\min}(|h_t|)$. Therefore, the LLR of a MSD in (\ref{LLRMSD}) can be approximated by
\begin{equation}
\label{maxlog}
\lambda_{n}\approx \frac{1}{N_0}[-\min_{x \in \mathcal{X}_{n}^0}|y-x|^2+\min_{x \in \mathcal{X}_{n}^1}|y-x|^2].
\end{equation}
This approximation, known as the max-log approximation (MLA) has been used for LLR estimation of Gray-mapped QAM \cite{Pyndiah1995}.  In \cite{Gul2011}, three methods, namely a MLA in (\ref{maxlog}), a log-separation algorithm (LSA) and a mixed algorithm, are proposed to reduce the complexity of the LLR calculation of a MSD for phase shift keying (PSK) and amplitude-phase-shift keying (APSK) constellations. In \cite{Gul2011}, it is shown that LSA is slightly better for low SNRs and for moderate-to-high SNRs the MLA is essentially the same as the exact calculation using (\ref{LLRMSD}). Further analysis of the MLA can result in substantial simplification of the LLR calculation. Here, by analyzing the MLA and employing the piecewise linear approximation of the LLRs, we reduce the LLR calculation complexity for PAM constellation with SPM.

\begin{proposition}
\label{Theorem1}
Let $\mathcal{X}_{n}$ be a $M_n=2^{B-n+1}$-point subset of a $2^B$-PAM constellation given that the $n-1$ upper-level (least significant) bits are known, and let $x_d$ be the $(d+1)^{th}$ largest element of $\mathcal{X}_{n}$. For $d \in {0,1,...,M_n-2}$ and $\mathbb{L}$ a large number, let $\Omega_d=\{y\in \mathbb{R} \enskip|-\mathbb{L}\delta_d+ x_d<y \leq x_{d+1}+\mathbb{L}\delta_{d-M_n+2}\}$, be the interval $(x_d,x_{d+1}]$ except the first interval extends down to $-\infty$ and the last up to $+\infty$. Then, for $y \in \Omega_d$
\begin{equation}
\label{maxlogPAM}
\lambda_{n}(y)\approx \frac{1}{N_0}(-1)^d(x_{d+1}-x_{d})(x_{d+1}+x_d-2y).
\end{equation}
\end{proposition}
\begin{proof}
Since $\mathcal{X}_{n}$ is just a $M_n$-PAM constellation with natural mapping that has been scaled by $2^{n-1}$ and shifted by an amount that depends on the known upper level bits, the label of the least significant bit changes with each consecutive point. Therefore, for $y \in \Omega_d$ the pair $(\min_{x \in \mathcal{X}_{n}^0}|y-x|^2,\min_{x \in \mathcal{X}_{n}^1}|y-x|^2)$ is either $(x_d,x_{d+1})$ or $(x_{d+1},x_{d})$ with the first applying when $d$ is even and the second when $d$ is odd. Thus (\ref{maxlog}) can be written as $\frac{1}{N_0}(-1)^d[-|y-x_d|^2+|y-x_{d+1}|^2]$, which reduces to (\ref{maxlogPAM}). 
\end{proof}

\begin{figure*}
\centering  
\subfigure[]{\label{fig:e}\includegraphics[width=54mm]{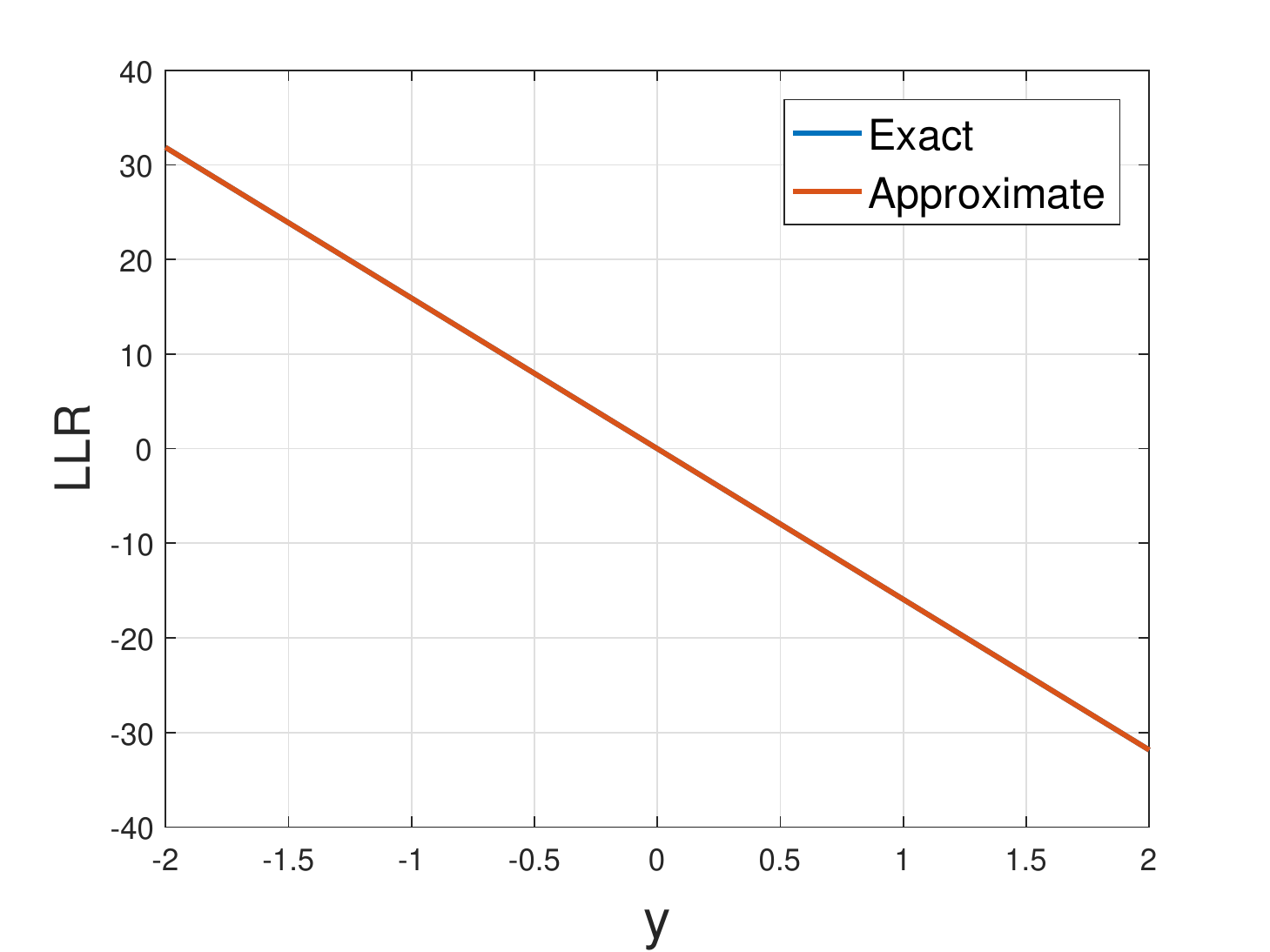}}
\subfigure[]{\label{fig:a}\includegraphics[width=54mm]{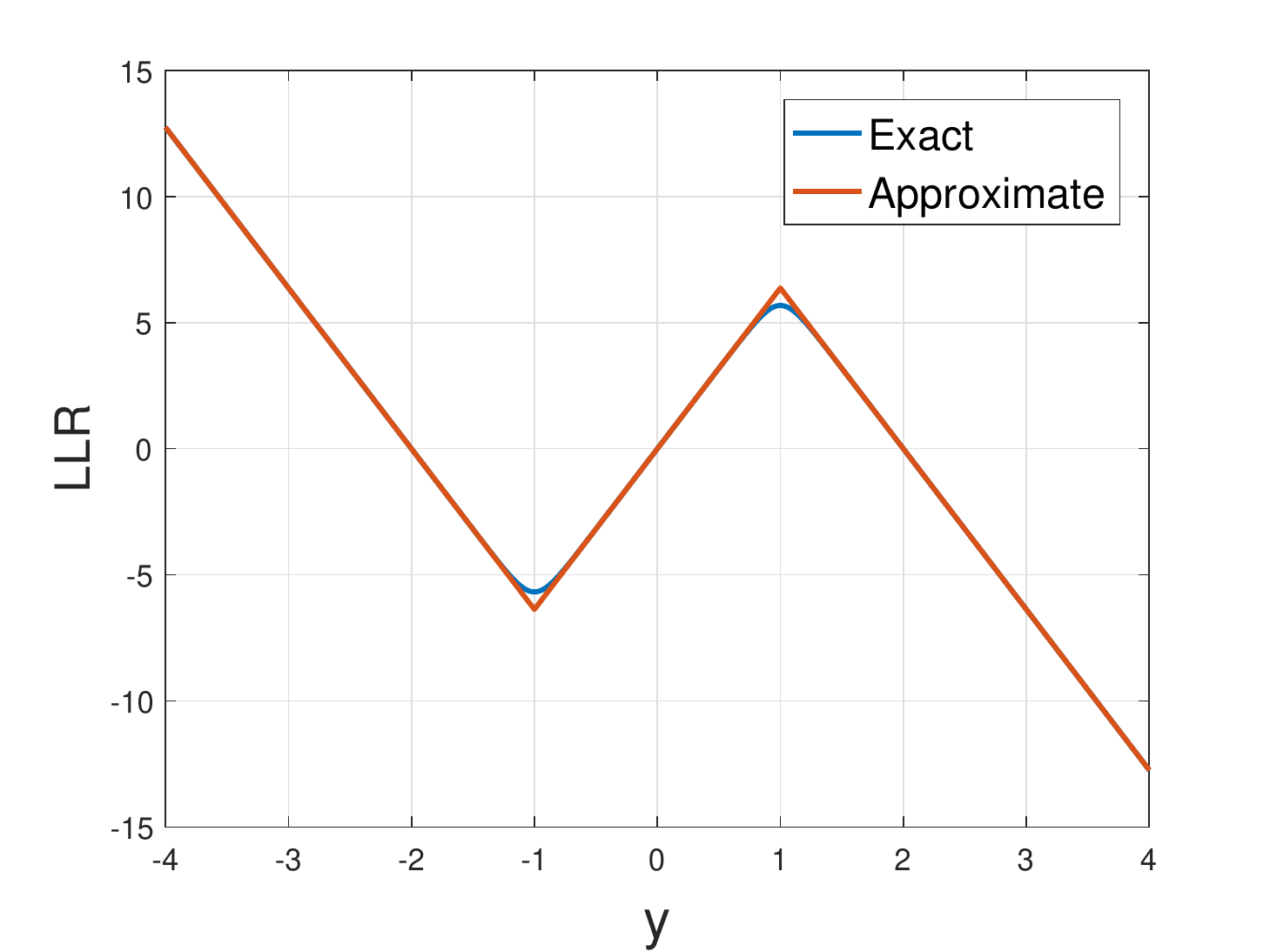}}
\subfigure[]{\label{fig:b}\includegraphics[width=54mm]{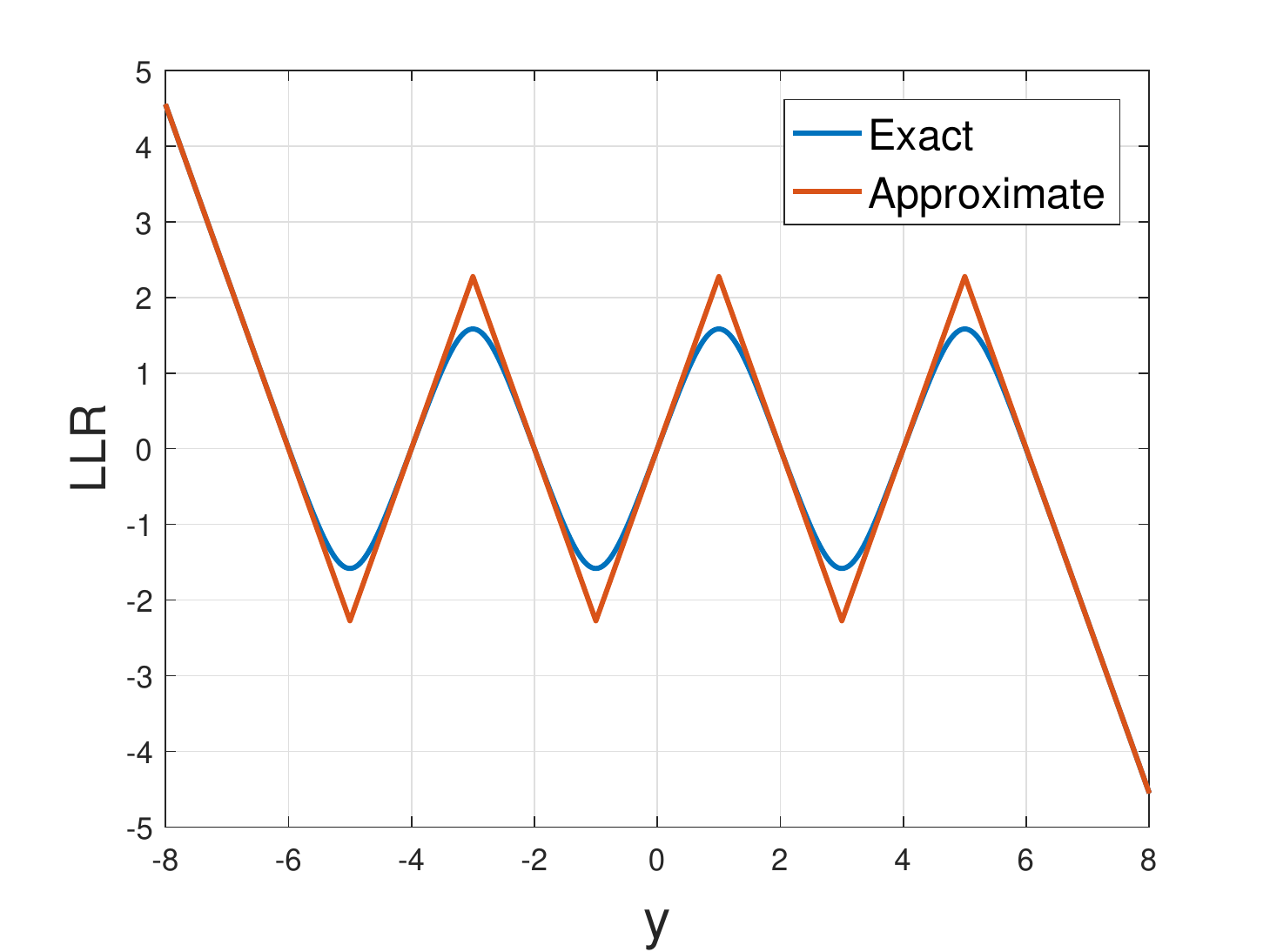}}
\caption{Comparison of the exact and the approximate estimation of LLRs as a function of $y$ for a) a BPSK, b) the first level of 4-PAM, and c) the first level of 8-PAM, respectively.}
\label{fig:PAMLLR}
\end{figure*}

The LLR as a function of the received symbol $y$, approximated by (\ref{maxlogPAM}) is compared with the exact LLR computed by (\ref{LLRMSD}) as shown in Fig.~\ref{fig:PAMLLR}, for the first level of 4-PAM, 8-PAM and BPSK with $\gamma=6$ dB. This approximation is quite accurate at SNRs of interest but require much less effort to evaluate than the exact expression.   Even thought this approximation is less accurate at low SNRs, we were unable to find a case where the use of the simplified LLRs had any discernible impact on the
throughput of the codes designed in this paper.

The proposed LLR approximation for QAM with 2D SPM consists of the max-log-based LLR approximation for two PAMs with SPM and a linear device. This can be computed using $9B-8$ real additions, $5B-6$ real multiplications, $B$ real divisions, $2B$ exponential operations, and $B$ logarithm operations. In contrast, if we estimate the LLR directly using (\ref{LLRMSD}), we need $8\times2^B-7-2B$ real additions, $6\times2^B-7$ real multiplications, $B-1$ real divisions, $2\times2^B-4$ exponential operations, and $B-1$ logarithm operations. Note that the complexity of the simplified LLR approximation grows with $B$ while the exact estimation grows with $2^B$. The proposed LLR simplification method is less complex when $B\geq2$.

\section{Throughput as a Design Metric}
\label{sec:Throughput}

Typically, for the AWGN channel, polar codes are designed to minimize the FER for a given code rate, $R$, at a given SNR. That is, the $K$ elements of the information-set are chosen in an attempt to provide as low a FER as possible. Alternatively, for a given SNR and a target FER, one can choose the information-set to be as large as possible (thereby maximizing the code rate), while ensuring the target FER is not exceeded. However, for systems employing HARQ, neither the FER nor the code rate is of primary importance. For ARQ systems, messages are transmitted indefinitely until they are correctly received. As such, the throughput, which is the rate that message bits are correctly received (in message bits per channel use), is a more relevant metric. By using this new criterion, codes that are quite different from  those that maximize the code rate or minimize the FER can be designed.

For the NC-D protocol, where a single code is used across all binary channels of the constellation and received samples from previous transmission are not exploited, the throughput is given by 
\begin{equation}
\label{throughput}
\eta_{\text{NC-D}}= \frac{K}{NB} (1-P_{K}),
\end{equation}
where $P_K$ is the FER when the message word length is $K$. The optimization problem, for a given codeword length, $N_{\text{tot}}=NB$, is to find $K$, and the associated information-set of the polar code, that maximize $\eta_{\text{NC-D}}$. When the NC-I protocol is used, the throughput, since the levels are independent, is 
\begin{equation}
\label{throughputID}
\eta_{\text{NC-I}}= \sum_{n=1}^{B} \frac{K_n}{N} (1-P_{n}^{K_n}),
\end{equation}
where $K_n$ is the length of the information-set of the codeword transmitted on the $n^{th}$ binary channel of the constellation (with $K=\sum_n K_n$ still being the total number of transmitted message bits.) and $P_n^{K_n}$ is the corresponding FER of that code.

When Chase-combining is used, the FER decreases depending on the number of retransmissions, so the throughput of the CC-D protocol is

\begin{equation}
\label{throughputcc}
  \eta_{\text{CC-D}}=  \frac{K}{NB \displaystyle \sum_{l=1}^{L} \prod_{l^{\prime}=1}^{l-1}P_{l^{\prime}}^{K}}(1-P_L^K),
\end{equation}
where $P_l^K$ is the FER of the $l^{th}$ transmission, given that the previous $l-1$ transmissions of the codeword failed. For CC-I, the throughput is 

\begin{equation}
\label{throughputccI}
  \eta_{\text{CC-I}}=  \sum_{n=1}^{B} \frac{K_n}{\displaystyle N \sum_{l=1}^{L} \prod_{l^{\prime}=1}^{l-1}P_{n,l^{\prime}}^{K_n}}(1-P_{n,L}^{K_n}).\;
\end{equation}

\section{Polar Code Design Methods for BPSK}
\label{polarcodedesign}

In this section, we start by explaining the simulation-based design method for SCD. Then, we describe the design method for NC and CC HARQ using GA which can design the code with low complexity.

\subsection{Simulation-based Code Design for SCD}
\label{polarcodedesignSCsim}

To design polar codes, the positions of the information bits must be determined. Determining the information-set by using Monte Carlo simulation, proposed by Arikan in \cite{Arikan2009}, is one of the methods of polar code design which benefits from high flexibility for adapting to a variety of practical channels.  In the simulation-based design method, as described in  \cite{Balogun2016}, the transmission of a large number of message words is simulated and SCD  decodes bits subsequently from the first to the last.  Then, the number of the first error events\footnote[2]{For each codeword, the first error event defined as the first erroneous output bit. This error does not include the propagated error and just represents the error of the bit-channel.} for each bit-channel is measured. The number of transmitted codewords for achieving the sufficient statistic can be decreased if, after recording each first error event, the corresponding bit is corrected to prevent propagating that error and the next bit-channels are examined subsequently. When the polar code is designed for a predetermined rate $R$ at a specific SNR, the best information-set is chosen to minimize the FER by finding the $K$ message bit positions with the lowest number of the first error events.

By recording the position of each first error event for each simulated codeword, it is easy to evaluate the FER for any information-set. Any simulated codeword with at least one first error event in positions specified by that information-set would also have been decoded incorrectly by a real decoder for the polar code defined by that information-set. Thus, the FER for a given information-set can be approximated  by dividing the number of incorrectly decoded codewords by the number of simulated codewords.

More formally, suppose we simulate $N_{\text{SIM}}$ codewords of length $N$ at a given SNR. Let $\epsilon_{m,\kappa}=1$ if a first error event occurred in the $\kappa^{th}$ bit-channel during the $m^{th}$ simulated codeword transmission, and  $\epsilon_{m,\kappa}=0$ otherwise\footnote[3]{Let $Z_\kappa= \Sigma_{m=1}^{N_{\text{SIM}}} \epsilon_{m,\kappa}$ be the total number of first error events in the $\kappa^{th}$ bit-channel. The information-set, $\mathcal{A}$, of the minimum-FER code of rate $K/N$ contains the values of $\kappa$ with the $K$ smallest values of $Z_\kappa$.}. For a given hypothetical information-set, $\mathcal{A}$, the $m^{th}$ simulated codeword would have been decoded incorrectly if $\Sigma_{\kappa \in \mathcal{A}} \, \epsilon_{m,\kappa}>0$. Let $\xi_m=1$ if $\Sigma_{\kappa \in \mathcal{A}} \, \epsilon_{m,\kappa}>0$. Out of the $N_{\text{SIM}}$ simulated codewords, the number of codewords that would have been incorrectly decoded is $\Sigma_{m=1}^{N_{\text{SIM}}} \xi_m$.

Using this method, it is straightforward to design polar codes to maximize the throughput. Once the simulation of a sufficiently large number of codewords has completed (typically $N_{\text{SIM}}=10000$ codewords is sufficient for low-to-moderate SNRs) at the desired SNR and the position of the first error events has been recorded ($\epsilon_{m,\kappa}$), the information-set of the minimum FER polar codes for every code rate from $\frac{1}{N}$ to $1$ is determined (i.e., $\forall K \in \{1,...,N\} $) and the corresponding FER is approximated. Then, the code rate that maximizes the throughput, (\ref{throughput}), is determined, and the associated information-set is used to define the optimal polar code at that SNR.

\subsection{GA-based Code Design for SCD}
\label{polarcodedesignSCGA}

As shown in \cite{Trifonov2012}, polar codes can be designed with low complexity using the GA. In GA, the LLR distributions for each node in polar graph can be approximated with a Gaussian distribution, in which the variance of the distribution $\sigma^2$ is two times the average LLR of a specific bit-channel. Therefore, to implement the GA-based design for polar codes we only need to update the average LLR through the polar graph. The updating average LLR rule for upper bit-channels is $\phi^{-1} (1-[1-\phi(\bar{\lambda_1})][1-\phi(\bar{\lambda_2})])$ and for lower bit-channels is $\bar{\lambda_1}+\bar{\lambda_2}$ where $\bar{\lambda}=E[\lambda_1]$ \cite{Vangala2015}. From \cite{Chung2001}, $\phi(h)$ is approximated as 
\begin{equation}
\label{GAestimate}
\phi(h)\triangleq\begin{cases} 1 & h=0, \\ \text{exp}(-0.4527h^{0.86}+0.0218) &  0<h \leq 10, \\ \sqrt{\frac{\pi}{2}}(1-\frac{10}{7h})\text{exp}(\frac{-h}{4}) & 10<h.\end{cases}
\end{equation}

Here, we design the throughput-optimal polar codes using the GA and compare it with the simulation-based approach. For designing the code we call Function~\ref{Algorithm21}(GA-BER($\gamma$,$N$),$N$) in Appendix~\ref{sec:appendix}. Function~\ref{Algorithm21} approximates the FER of the code based on a well-known bound given as \cite{Seidl2013-1}
\begin{equation}
\label{FERbound0}
P_{K}=1-\displaystyle  \prod_{i=1}^{K}(1-v_i),
\end{equation}
where $v_i$ is the bit error rate (BER) of the $i^{th}$ bit-channel given previous bit-channels are frozen. These BERs are estimated using Function GA-BER in Appendix~\ref{sec:appendix}. In the next step, the corresponding throughput for each message-length is computed and the code rate with the highest throughput is chosen. From (\ref{throughput}) and (\ref{FERbound0}), the throughput can be written as 
\begin{equation}
\label{FERbound1}
\eta_{\text{NC},K}=\frac{K}{NB}\displaystyle  \prod_{i=1}^{K}(1-v_i).
\end{equation}

\begin{proposition}
\label{Lemma1}
The FER estimation in $P_{K}$ is monotonically increasing with respect to $K$.
\end{proposition}
\begin{proof}
The FER estimation in (\ref{FERbound0}) can be written in a recursive form as
\begin{equation}
\label{FERsplit}
\begin{split}
P_{K}= 1- (1-v_K)\prod_{i=1}^{K-1}(1-v_i) = P_{K-1}+v_K \prod_{i=1}^{K-1}(1-v_i).
\end{split}
\end{equation}
Thus, (\ref{FERsplit}) monotonically increases with $K$ since $v_K \prod_{i=1}^{K-1}(1-v_i)>0$ due to $0<v_i<1$.
\end{proof}

\begin{proposition}
\label{Theorem55}
The throughput in (\ref{throughput}) is a unimodal function of the code rate. 
\end{proposition}
\begin{proof}
The first forward difference can be given as
\begin{equation}
\label{FERsplit1}
\begin{split}
&\Delta \eta_K=\eta_{K+1}-\eta_K =\frac{1}{NB}\big[1-(K+1)v_{K+1}\big]\prod_{i=1}^{K}(1-v_i).
\end{split}
\end{equation}
Since $\Delta \eta_K>0$ for $K<\frac{1}{v_{K+1}}-1$ and $\Delta \eta_K<0$ for $K>\frac{1}{v_{K+1}}-1$, the difference equation has only one root, so $\eta_K$ has only one maxima. Thus, (\ref{throughput}) is unimodal.
\end{proof}

The throughput vs. the code rate for polar codes with length 4096 at SNRs of -2, 0 and 2 is shown in Fig.~\ref{fig:Throughput_change}. From Proposition~\ref{Lemma1}, we know $P_{K}>P_{K-1}$. Thus, the throughput initially grows (nearly linearly for the FER-minimal codes) with the code rate until the rate gets sufficiently high that the effects of the FER start to dominate in (\ref{throughput}), after which point the throughput drops dramatically. The existence of an optimal rate to maximize the throughput is clear which is shown in Proposition~\ref{Theorem55}.
\begin{figure}[b]
  \center
   \includegraphics[width=0.5\textwidth]{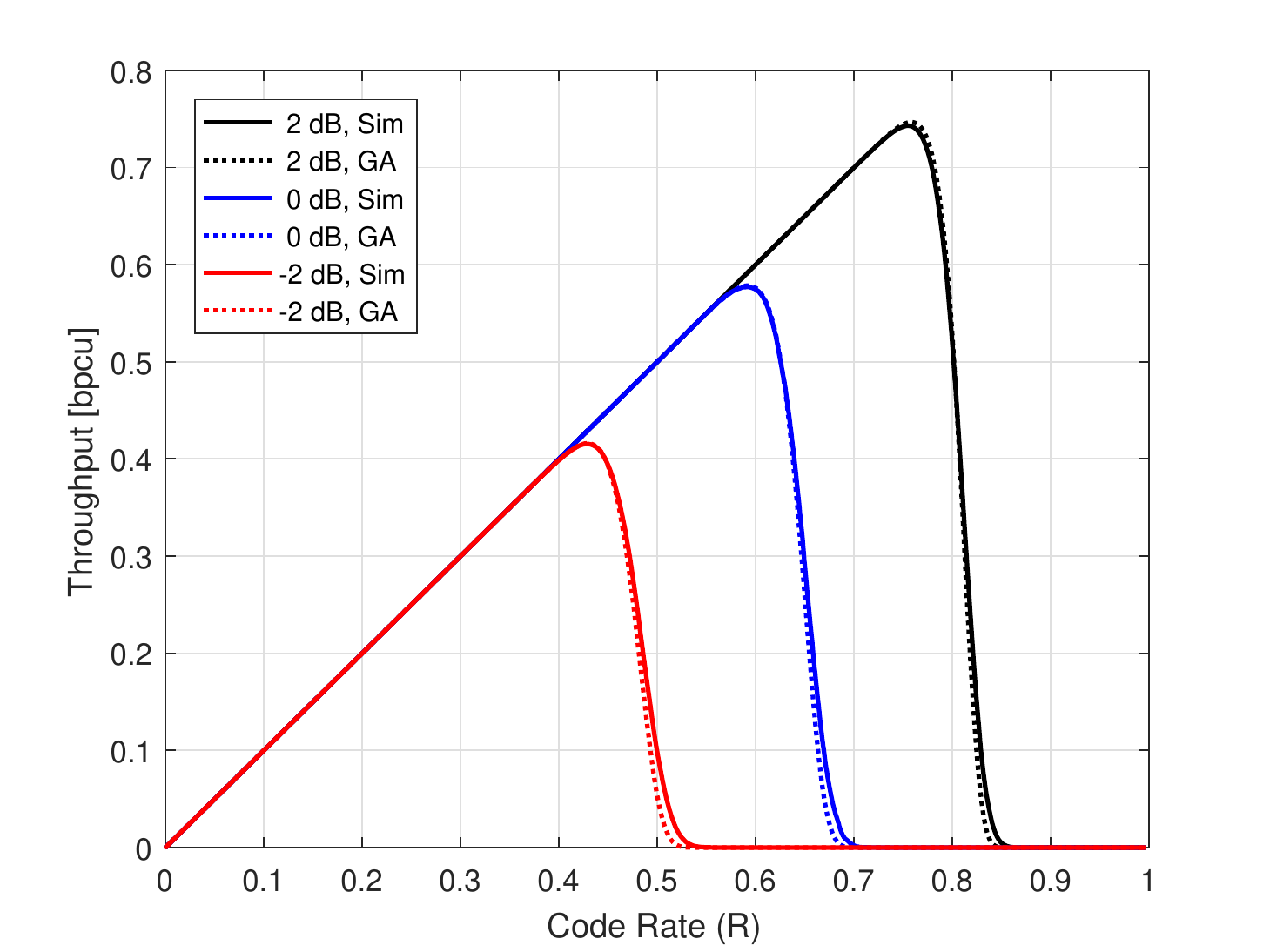}
\caption{Throughput of NC protocols vs. the code rate with $N=4096$ bits at different SNRs for BPSK with SCD.}
\label{fig:Throughput_change}
\end{figure}

\subsection{Code Design for Chase-combining}

To design the code for the CC HARQ schemes, the code design method for NC should be modified since the average LLR with each retransmission increases to $\frac{4l}{N_0}$, where $l \in \mathbb{N}$ is the retransmission attempt number. In this case, the procedure of the design includes running the Function GA-BER in Appendix~\ref{sec:appendix} to estimate the order of bit-channels for the first transmission since the majority of codewords are decoded correctly in the first step. To find the rate corresponding to the maximum throughput, the FERs using the Function~GA-BER for each retransmission are estimated and (\ref{throughputcc}) is computed. For the stopping criterion, one can use a threshold on the probability of error of the maximum message length, e.g., $P_L^N<10^{-6}$. However, this criterion converges slowly. As a fast criterion, the lack of change of the maximum throughput is employed in this paper. The code design for CC HARQ is described formally in Function~\ref{Algorithm33} in Appendix~\ref{sec:appendix}. To design the code, we should call Function~\ref{Algorithm33}($\gamma$,$N$). The codes designed for CC have slightly higher rates than the codes designed for NC.

\subsection{Code Design for Incremental Redundancy}
\label{polarcodedesignSCDIR}

CC has the advantage over NC that the received sample from all the retransmissions are used when making decisions. IR schemes have the additional advantages that each retransmission does not have to be of the same length, and can belong to an extended code instead of merely repeating code bits. A very good incremental redundancy scheme based on punctured polar codes was recently proposed in \cite{Yuan2018}. In this scheme, a family of $L$ polar codes of lengths $N_1<N_2<...<N_L$ are designed for a fixed message word of $K$ bits. By using quasi-uniform puncturing \cite{Niu2013} along with careful repetition of message bits at the encoder and some minor modifications to the polar decoder (see \cite{Yuan2018} for details), these codes are rate-comparable. The throughout-maximizing design technique presented here can easily be extended to this IR scheme. Although better performance can be achieved by using a small incremental transmission size (e.g., one bit at a time), to simplify the scheduling and to provide a meaningful comparison to the NC and CC schemes, in the following we consider an increment of $N$ bits (i.e., $N_l=lN$). The design procedure is the same as described as in Function~\ref{Algorithm33}, except $\bar{\lambda}=4\times 10^{\gamma/10}$ is used in line 3 and $v=$GA-BER($\bar{\lambda},N_l$) is used in line 4 and the BER sorting is repeated for all retransmissions.

A comparison of the throughput of NC, CC and IR protocols is provided in Fig~\ref{fig:Throughput_changeNCCCIR}. The IR scheme is designed for $L=4$. Observe that the IR code achieves the throughput of NC in the first transmission since the information set of codes are the same, and the throughput of retransmissions shows a higher peak.

\begin{figure}[t]
  \center
   \includegraphics[width=0.5\textwidth]{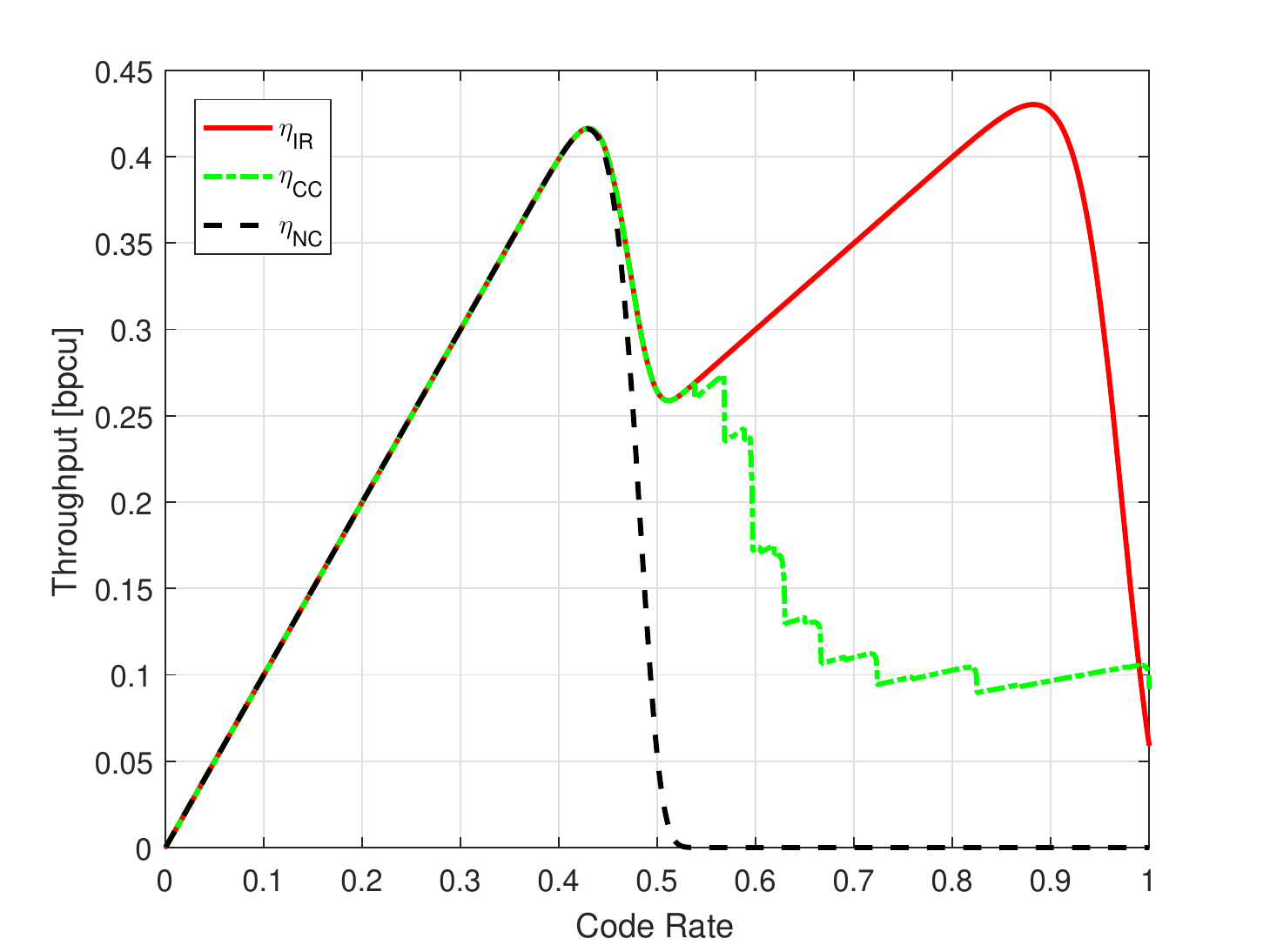}
\caption{Throughput comparison of NC, CC, and IR protocols vs. the code rate with $N=4096$ bits and SCD for BPSK at $\gamma=-2$ dB.}
\label{fig:Throughput_changeNCCCIR}
\end{figure}

\section{Polar Code Design Methods for QAM}
\label{polarcodedesignQAM}

To design MLPCM for PAM and QAM constellations, the GA method can be employed. In this section, we explain the design steps of MLPCM and provide algorithms  to maximize the throughput of HARQ schemes.

\subsection{GA-based Code Design for SCD}
\label{polarcodedesignSCQAM}

To construct polar codes using GA, the average LLR for each binary channel of a QAM constellation is estimated and an independent binary polar code is constructed for each level. Due to the Gaussian noise, the LLR distribution for a BPSK constellation in an AWGN channel is $\mathcal{N}(\frac{4}{N_0},\frac{8}{N_0})$. However, the LLR distributions of each level of a multilevel coding scheme for PAM and QAM constellations are not known. The average LLR of QAM with $2^B$ points can be estimated using the average LLRs of the constituent PAM with $2^{B/2}$ points based on the transform described in Theorem~\ref{Theorem11}. In this theorem, we indeed use a linear block code to estimate the LLRs of QAM with 2D SPM from LLRs of PAM. This linear block code is similar to the second stage of a polar encoder and hence, a modified GA can be employed to estimate the average LLRs of QAM. The corresponding average LLR estimation for $b=1,...,B/2$ can be expressed as
\begin{equation}
\label{PAMLLRn1}
\begin{split}
& \bar{\lambda}_{\text{QAM},2b-1}=\phi^{-1} (1-[1-\phi(\bar{\lambda}_{\text{PAM},b})]^2), \\
& \bar{\lambda}_{\text{QAM},2b}=2\bar{\lambda}_{\text{PAM},b}.
\end{split}
\end{equation} 

To estimate the average LLRs of a PAM constellation, we use the piecewise linear estimation in Proposition~\ref{Theorem1}. Consequently, since the LLR is a linear function of $y$, the average LLR of each level of PAM, given a zero is transmitted, can be computed through integration. For the first level of PAM, the integral is given by
\begin{equation}
\label{avgLLRint}
\bar{\lambda}_{1}=\sum_{d | x_d \in \mathcal{X}_{1}^0} \int_{\Omega_d}\Pr(x_d) \lambda_1(y) \frac{1}{\sqrt{\pi N_0}} e^{\frac{-(y-x_d)^2}{N_0}}dy, 
\end{equation}
where  $\Pr(x_d)$ is the probability of the transmission of $x_d$ within $\mathcal{X}_{1}^0$. For example for 4-PAM, $\Pr(x_d)=0.5$ and for 8-PAM, $\Pr(x_d)=0.25$. For other levels, the integration is taken over the corresponding set $\mathcal{X}_{n}^0$. Therefore, the general form of (\ref{avgLLRint}) can be written as
\begin{equation}
\label{avgLLRint2}
\bar{\lambda}_{n}=\sum_{d | x_d \in \mathcal{X}_{n}^0} \int_{\Omega_d}\Pr(x_d) \lambda_n(y) \frac{1}{\sqrt{\pi N_0}} e^{\frac{-(y-x_d)^2}{N_0}}dy, 
\end{equation}
where $\lambda_n(y)$ is given by the LLR approximation in (\ref{maxlogPAM}). As a numerical example, for 4-PAM and 8-PAM, the average LLRs of  the binary channels are $[6.3,31.9]$ and $[0.7,6.0,30.5]$ at $\gamma=10$ dB, respectively. 

In \cite{Trifonov2012}, it is noted that the FER of a MLPCM can be approximated with  $1-\prod_{n=1}^{B}(1-P_n)$. By extending this bound, the total FER of the MLPCM can be approximated as
\begin{equation}
\label{FERbound1}
P_{K}=1-\displaystyle \prod_{n=1}^{B} \prod_{i=1}^{K_n}(1-v_n^i),
\end{equation}
where $v_n^i$ is the BER of the $i^{th}$ bit-channel in $n^{th}$ level given previous bit-channels are frozen.

To find the level code rates, \cite{Trifonov2012} suggests using the equal error probability rule in which all levels of a MLPCM have approximately the same FER. However, this requires solving a program to find code rates. From (\ref{FERbound1}), one can observe that the MLPCM works like a longer single binary polar code while it observes a variety of equivalent SNRs corresponding to different levels. Thus, to determine level code rates, the bit-channel reliabilities $v_n^i, \forall i=1,...,N$ can be measured for all levels $n=1,...,B$ and among them, those with the lowest genie-aided BERs are chosen as the total information-set \cite{Seidl2013-1}. This automatically determines the rate of each level since some bit-channels of each level are in the total information-set. When we use this rule to design the code, the FERs of different levels are very close. However, this rule does not require solving any program to determine the code rates which highly simplifies the MLPCM design. The entire code design procedure for the NC-D protocol is mentioned in Algorithm~\ref{Algorithm5} provided in Appendix~\ref{sec:appendix}.

An algorithm for designing throughput optimal codes for NC-I protocol is presented in Algorithm~\ref{Algorithm7} in Appendix~\ref{sec:appendix}. Due to the independence of levels, the algorithm designs a binary code for each level independently. The same algorithm can be developed for the CC-I protocol by substituting the Function~\ref{Algorithm33} instead of the Function~\ref{Algorithm21} in Algorithm~\ref{Algorithm7}. Similar procedure can be used to design the IR-I protocol based on the code design method explained in Section~\ref{polarcodedesignSCDIR} for BPSK.  Due to the space limitation, we do not mention the design algorithm for the CC-D and IR-D protocols but it is only the extension of Function~\ref{Algorithm33} and Algorithm~\ref{Algorithm5}.

\section{Rate Matching Algorithm for SCLD}
\label{polarcodedesignSCL}

For decoding of each output bit with SCD, the information of other previously decoded bits and the future information bits are not used. To overcome these shortcomings, the SCLD records a list containing different possible decoded message words and keeps only $L$ most likely ones after each step \cite{Tal2015}. A CRC sequence is usually added to message bits when SCLD is used, to increase the probability of finding the most likely message word. Throughout this paper for SCLD, only one CRC sequence is used for both list decoding and ARQ. Typically, the codes designed for SCD are used for SCLD as well since the SCL core decoder is SCD. However, these codes are suboptimal for SCLD. 

When throughput-maximizing codes optimized for SCD are used with SCLD, the FER is lower than the FER of SCD. Even though this slightly improves the throughput, it is not highly effective  on the term $(1{-}P_{K})$ in (\ref{throughput}). However, since $R$ is numerically more dominant in (\ref{throughput}) when the FER is small, it can be increased more significantly to improve the throughput. Therefore, we introduce a rate matching algorithm for SCLD. The algorithm employs the golden-section search method to find the code rate corresponding to the maximum throughput. 

The golden-section search method iteratively measures the objective function at different points and updates the answer range interval $[a,b]$ until this interval is narrowed down around the final value of the decision variable. For a detailed explanation of the golden-section search method refer to \cite{Thisted1988} and references therein. Here, the objective function is the actual throughput of SCLD measured using simulation and the decision variable is the message word length. The proposed algorithm is fast, e.g., for $N=16384$ it finds the optimum rate in around 15 iterations, corresponding to 16 evaluations of the objective function. For initialization of the algorithm, we use $a=K_{\text{SCD}}$, the message word length for SCD, and $b=\text{min}(a+\frac{BN}{10},BN)$. The algorithm is formally presented in Function~\ref{Algorithm4} in Appendix~\ref{sec:appendix}. Note than when we use the algorithm for the NC-D protocol, we should call Function~\ref{Algorithm4}($\gamma$,$K$,$N$,$\textbf{idx}_{\text{tot}}$,$B$) to repeat the simulation for all $B$ binary channels of a constellation. However, for the NC-I protocol we call  \ref{Algorithm4}($10\log_{10}\frac{\bar{\lambda}_n}{4}$,$K_n$,$N$,$\textbf{idx}_n$,1) for each binary channel $n=1,...,B$, independently. 

A comparison of throughput vs. the code rate for SCD and SCLD at different lengths is presented in Fig.~\ref{fig:th3scld}. Observe that the rate-matched codes for SCLD substantially improve the throughput for all code lengths. Furthermore, as we increase $N$, the throughput tends to the capacity at which the throughput and the code rate of the throughput-maximizing code eventually are equal.  

\begin{figure}[t]
\center
    \includegraphics[width=0.5\textwidth]{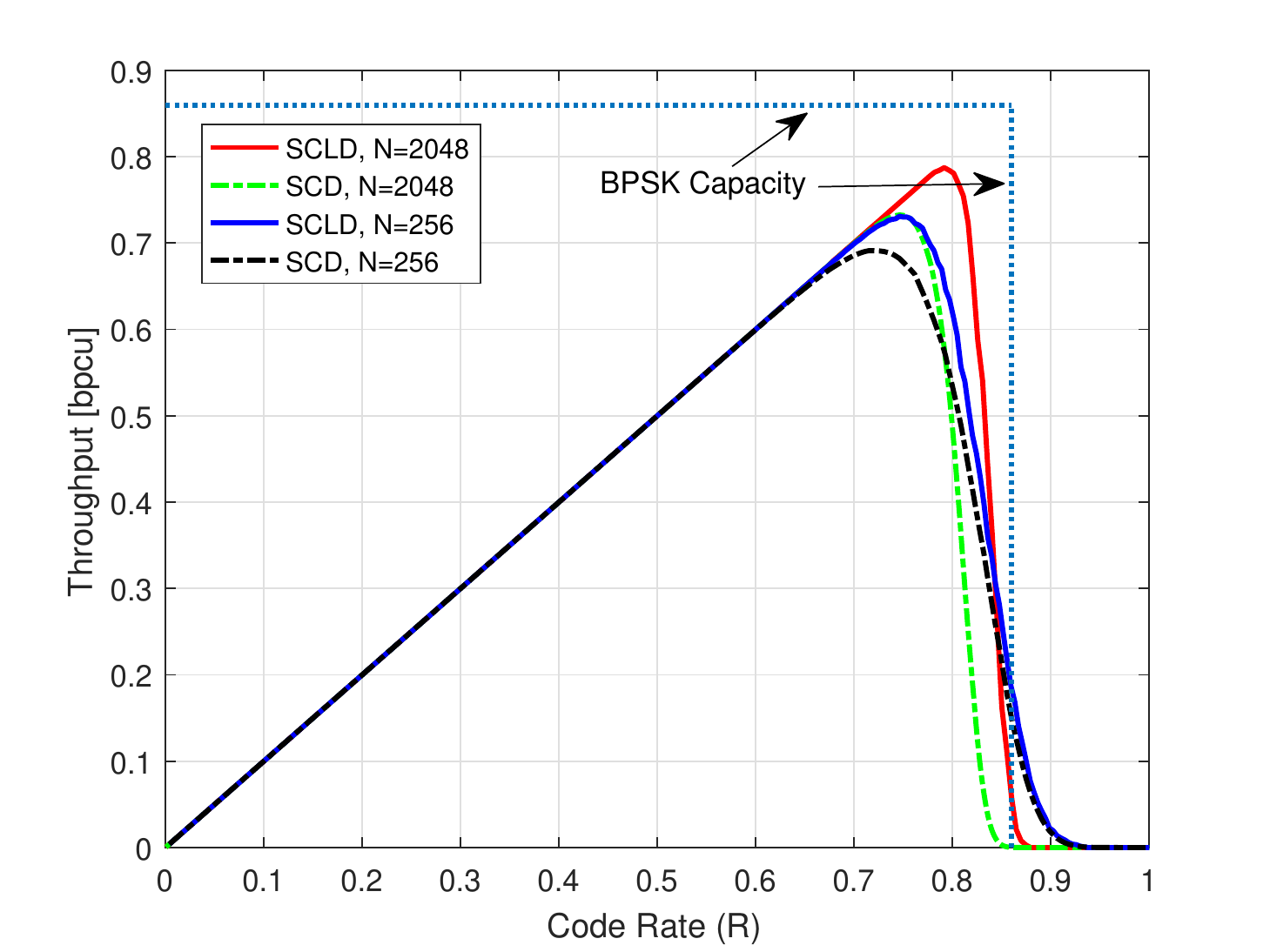}
\caption{Throughput of NC protocol vs. the code rate with different code lengths at $\gamma=2$ dB for BPSK.\hfill \vspace{0.4cm}}
\label{fig:th3scld}
\end{figure}

\section{Numerical Results and Discussions}
\label{sec:num}

In this section, we provide the performance analysis of the code design algorithms described in Sections \ref{polarcodedesign} and \ref{polarcodedesignQAM}, respectively. The system described in Section \ref{sec:sysmodel} is used for all simulations and the CRC sequence is CRC-16-CCITT. 

The throughput of SC- and SCL-decoded polar codes of length 4096 vs. SNR ($E_s/N_0$) is shown in  Fig.~\ref{fig:result2} in comparison to BPSK capacity. The SCLD list size is 32 for all curves. The lowermost black curve shows the throughput of the polar code designed using SCD and decoded with SCD that achieves the throughput of 80\%  of the capacity at 0 dB. The second black curve is the throughput of the code designed for SCD and decoded using SCLD which achieves  82.5\%  of the capacity at 0 dB. The topmost curve under the capacity shows the performance of the code designed for SCD and rate matched for SCLD   which achieves the 89.3\% of the capacity at 0 dB. Therefore, the use of SCLD for decoding of codes designed for SCD does not change the throughput substantially in comparison to SCD. However, employing  Function~\ref{Algorithm4} for the rate matching can substantially improve the throughput of the code used with SCLD.

\begin{figure}[b]
\center
    \includegraphics[width=0.5\textwidth]{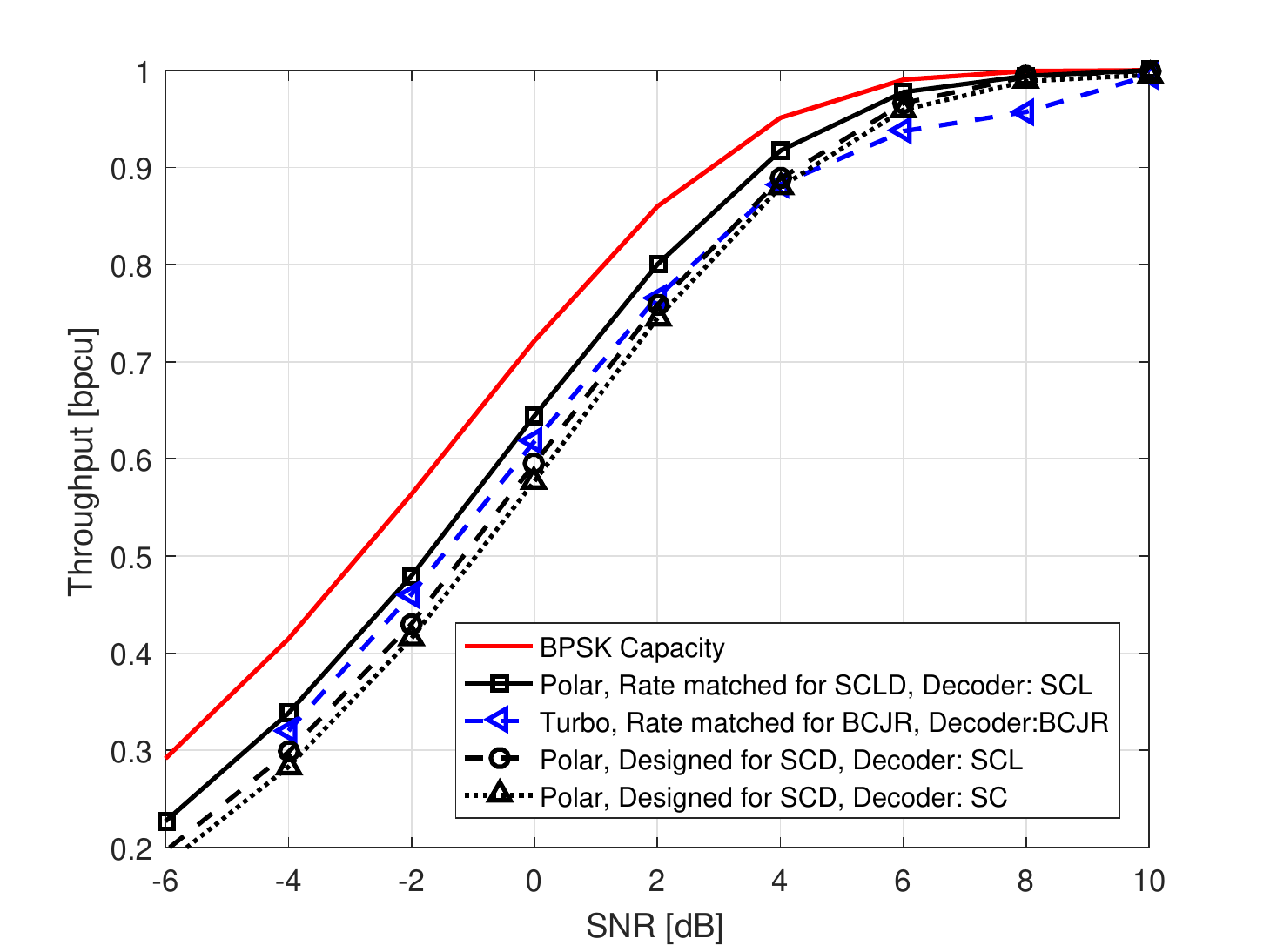} 
\caption{Throughput comparison of polar codes of length 4096 optimized for SCD decoded with SCD and SCLD, and the SCLD rate matched codes decoded with SCLD.\hfill \vspace{0.3cm}}
\label{fig:result2}
\end{figure}

Fig.~\ref{fig:result2} furthermore provides a comparison of the throughput-maximizing polar codes and parallel concatenated (turbo) codes employed in LTE-A \cite{ETSITS136212}. The BCJR decoder  with 5 iterations is employed for decoding of the turbo codes with codeword lengths of around 4096. The LTE-A turbo code rate is optimized using the golden-section search (similar to Function SCLD-Rate-Match but the BCJR is employed) to maximize the throughput. The range of message word lengths is limited to $40{:}8{:}512$, $528{:}16{:}1024$, $1056{:}32{:}2048$ and $2112{:}64{:}4200$ bits which provides us with 157 different choices for the code rate. In this case, a golden-section search is used to search all the possible 157 choices for the code rate and selects the code rate corresponding to the highest throughput. Due to code rate limitations, the optimization procedure was only applied in the SNR range between -4 and 10 dB. Furthermore, the turbo code lengths are slightly higher than 4096.  It can be observed that turbo code performance is close to the optimized polar code using SCLD-Rate-Match at low SNRs. However, as the SNR increases, the performance of turbo degrades and at high SNRs, it is even worse than the polar code optimized for SCD. Note that the complexity of BCJR with 5 iterations is more than that of SCLD.

The comparison of the CC HARQ  schemes proposed in \cite{Chen2014} for SCD and \cite{Liang2017} for SCLD and the CC HARQ design method introduced in Function~\ref{Algorithm33} is provided in Fig.~\ref{fig:result5} for BPSK. To construct the polar codes using algorithms provided in \cite{Liang2017} and \cite{Chen2014}, the $K$ at each SNR is considered the same as  $K$ found using Function~\ref{Algorithm33}. The results indicate that the methods proposed in \cite{Liang2017} and \cite{Chen2014} with at least $N-K$ times more complexity, cannot generate codes better than the Function~\ref{Algorithm33}. 

In addition, in Fig.~\ref{fig:result5}, the comparison of IR codes with CC and NC is also provided. To construct IR codes, the method described in Section~\ref{polarcodedesignSCDIR} is used. The NC codes optimized using Function~\ref{Algorithm21} have approximately the same performance as the codes generated for the CC using Function~\ref{Algorithm33} for situations where we can adapt the code when the SNR changes. In this case, CC does not have any sensible advantage over NC since when we design the codes for NC or CC, the algorithms try to minimize the number of retransmissions and keep it in order of at most one to maximize the throughput.  Observe that IR code outperforms NC with SCD and SCLD by 4.53\% and 3.82\% of the BPSK capacity at $\bar{\gamma}=-4$ dB, respectively. At moderate-to-high SNRs, the IR code cannot improve the throughput anymore.

\begin{figure}[b]
\center
    \includegraphics[width=0.5\textwidth]{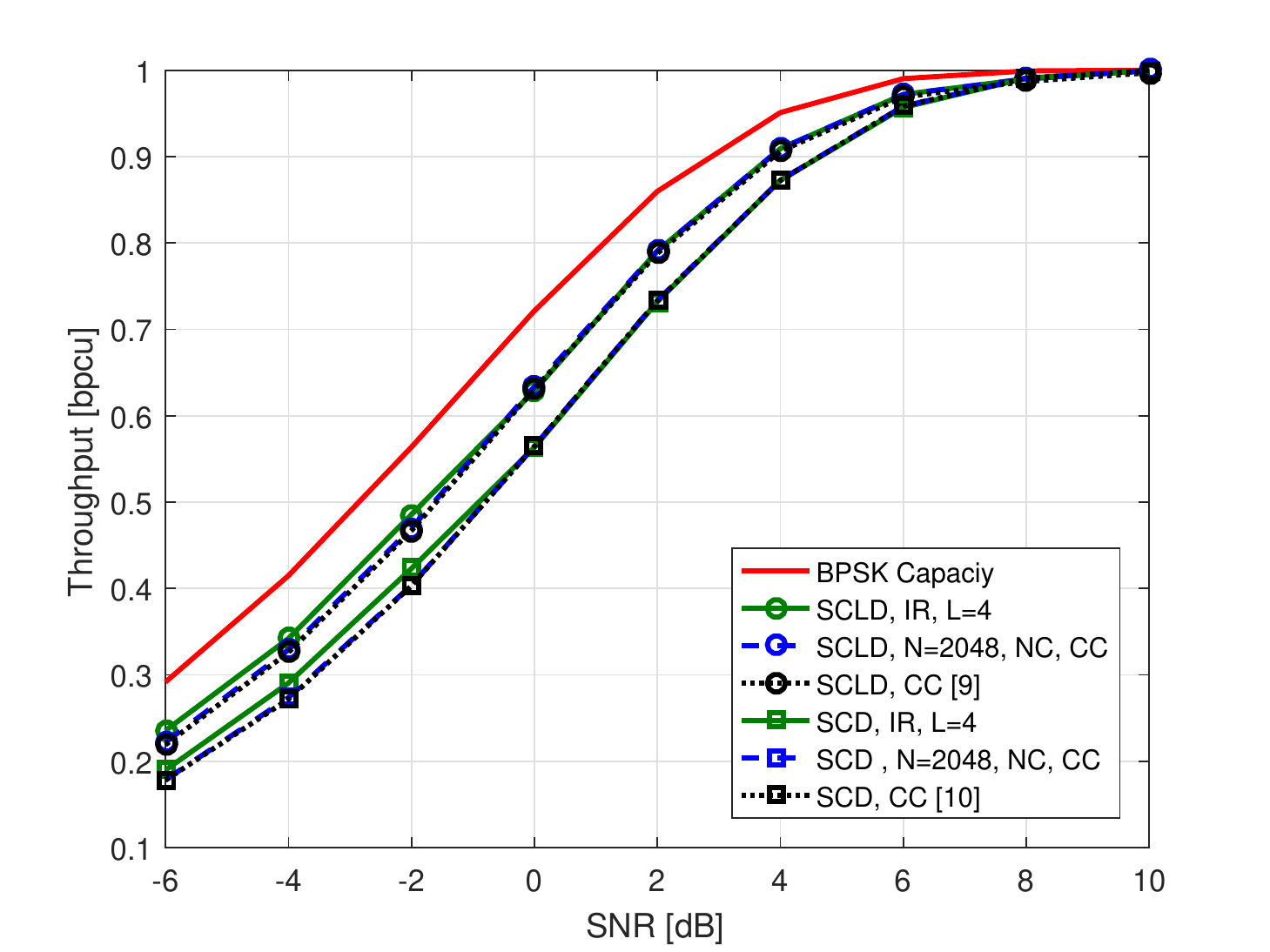}
\caption{Comparison of proposed NC, CC, and IR codes and code design methods in \cite{Chen2014} and \cite{Liang2017} for SCD and SCLD, respectively.\hfill
\vspace{0.3cm}}
\label{fig:result5}
\end{figure}

The comparison of NC-I and NC-D protocols with SCD and SCLD  is presented in Fig.~\ref{fig:result7} for 16-QAM. At all SNRs, NC-I and NC-D with SCLD perform slightly better than the corresponding protocol with SCD. At 4 dB, NC-D with SCD and SCLD achieves 69\% and 73\% of the capacity and NC-I with SCD and SCLD achieves 74\% and 84\% of the capacity, respectively. For all SNRs, the throughput of the code is essentially identical regardless of whether the exact LLR calculation of (2) or the simplified LLR calculation of Section~\ref{LLRSimple} is used. This is despite the fact that the simplified LLR approximation is not very exact at low SNR.

\begin{figure}[t]
\center
    \includegraphics[width=0.5\textwidth]{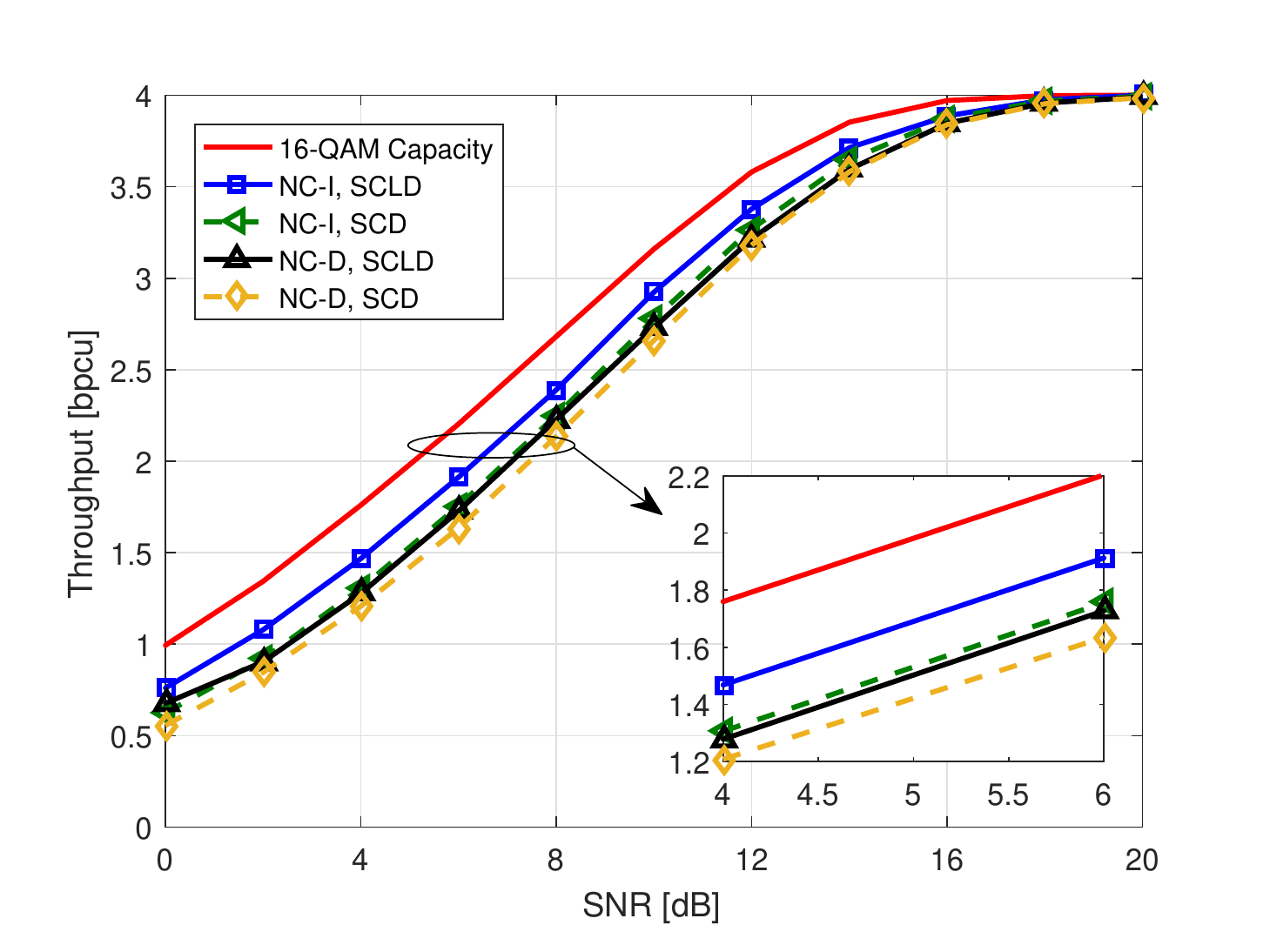}
\caption{Throughput comparison of  NC-D and NC-I protocols with SCD and SCLD for $N_{\text{tot}}=2048$ and 16-QAM.\hfill
\vspace{0.3cm}}
\label{fig:result7}
\end{figure}

In all previous results, we assumed we can adaptively change the code and the modulation. In those cases, CC HARQ protocols do not provide any advantage over NC protocols. However, once the number of codes is limited, the advantage of CC HARQ is highlighted. In Fig.~\ref{fig:result8}, the performance of CC-I protocol with two codes designed at 4 dB and 14 dB is compared with the IR HARQ and CC HARQ schemes constructed using polar codes and BICM in \cite{Elkhami2015} for 16-QAM. We observe that when we use the code designed at 14 dB for the whole range of SNR, IR HARQ is better than the proposed throughput-optimal codes for CC-I protocol at a few SNRs. However, on average the proposed MLPCM with CC-I protocol performs up to 3 dB better than IR HARQ and CC HARQ in \cite{Elkhami2015}. In case we use one more code designed at 4 dB, CC-I at all SNRs achieves higher throughput than schemes proposed in \cite{Elkhami2015}. The reason for this superiority is the good design of MLPCM scheme in conjugation with CC-I scheme. We also compared the MLPCM scheme with IR-I HARQ constructed based on the method explained in Section~\ref{polarcodedesignSCDIR}. Observe that it is only slightly better than CC at low SNRs.

\begin{figure}[t]
\center
    \includegraphics[width=0.5\textwidth]{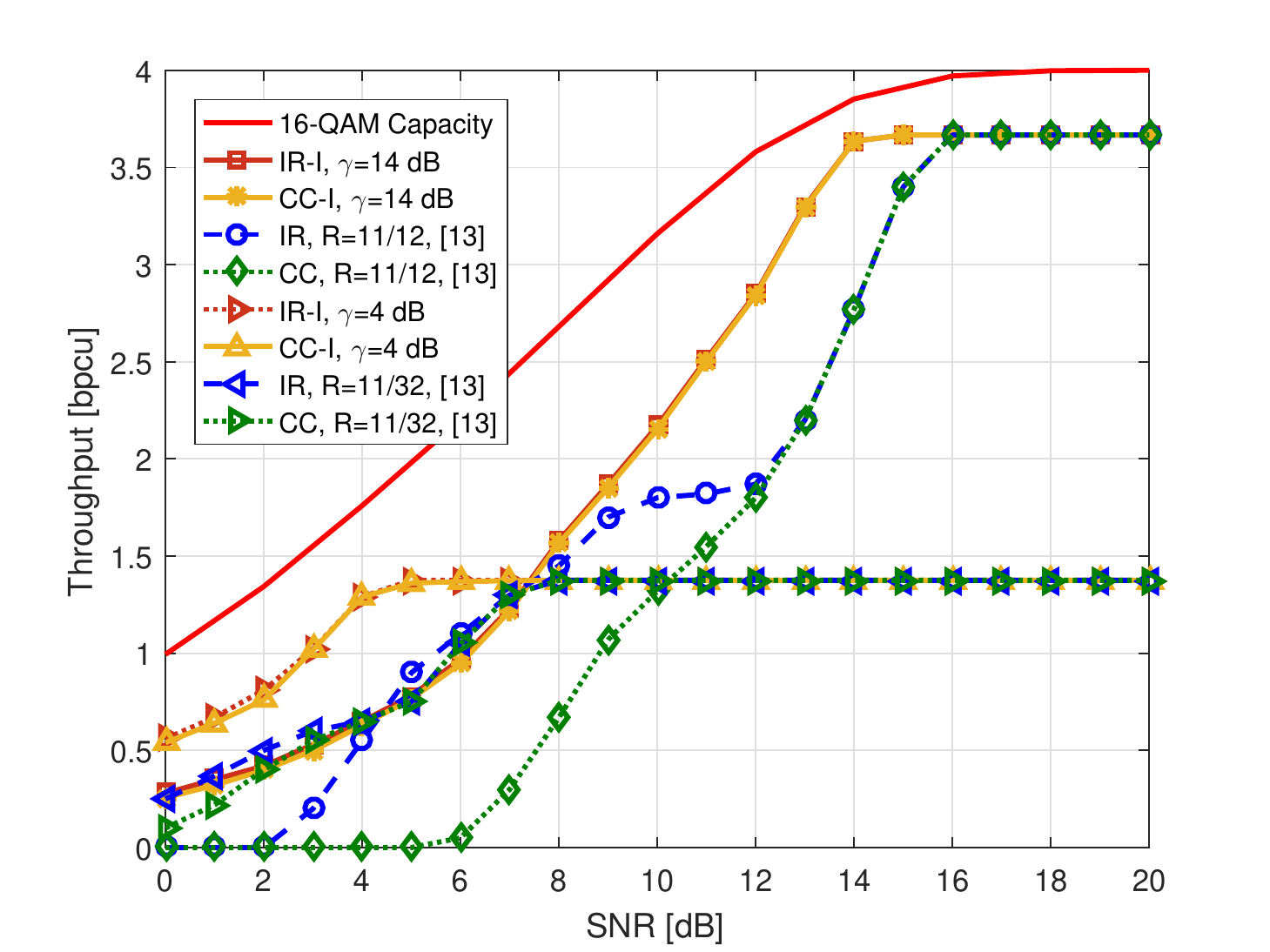}
\caption{Comparison of CC-I with two codes and the IR and CC HARQ proposed in \cite{Elkhami2015} for 16-QAM.\hfill
\vspace{0.3cm}}
\label{fig:result8}
\end{figure}

\section{Conclusion}
\label{sec:Conclusion}

In this paper, we simplified the accurate LLR estimation of QAM constellations with 2D SPM using a simple linear relationship and two PAM LLR estimators and we proposed a method for generating a SPM for QAM constellations with an even number of address-bits. We also simplified the LLR estimation of PAM constellations with SPM. Furthermore, we proposed a set of algorithms for designing MLPCM by maximizing the throughput for NC and CC HARQ schemes. The numerical results show the codes constructed using these methods perform very close to the capacity, e.g., within 1.2 dB of the 16-QAM capacity with $N=2048$ bits. The results indicate an improvement up to 15\% of the capacity at 4 dB when we use level-independent protocols and rate-matched SCLD. Due to the substantial throughput improvement that can be achieved using carefully designed codes, the idea of throughput-based polar code design should be extended to designing optimal codes for IR HARQ systems in the future.

\appendices
\section{Algorithms to generate codes}
\label{sec:appendix}

The algorithms used in this paper to generate codes are mentioned in this section.

\vspace{23pt}

\begin{function}[!ht]
 \DontPrintSemicolon
     \KwIn{$\textbf{v}$ and $N$}
    \KwOut{\textbf{v}: The sorted indices of the information-set $\textbf{idx}$ and the message length $K_{\text{opt}}$}
  \nonl \textbf{Procedures:} [\textbf{v},\textbf{idx}] =Sort(\textbf{v}): Sorts the input vector in an increasing order and outputs the sorted indices $\textbf{idx}$ and sorted values $\textbf{v}$.\;
         [\textbf{v},\textbf{idx}]=Sort($\textbf{v}$) \Comment{Sort the bit-channels}\;
       \For {$\kappa=1:N$} 
       {
        $P_{\kappa}=1-\displaystyle  \prod_{i=1}^{\kappa}(1-v_i)$ \Comment{Estimate FER for all $K$s}\;
        ${\eta}_\kappa=(1-P_{\kappa})\frac{\kappa}{N}$ \Comment{Estimate $\eta_\kappa$}\;
       } 
      $K_{\text{opt}}= \displaystyle \argmax_{\kappa} \eta_\kappa$ \Comment{Find $K$ with highest $\eta_\kappa$}\;
       $R=\frac{K_{\text{opt}}}{N}$ \;
       return $\textbf{idx}$, $K_{\text{opt}}$\;
           \caption{NC-Binary-Design($\textbf{v}$,$N$)}
        \label{Algorithm21}
\end{function}

\vfill\null

\begin{function}[!ht]
 \DontPrintSemicolon
    \KwIn{$\bar{\lambda}$ and $N$}
    \KwOut{\textbf{v}: The vector of the BERs of bit-channels}
    $\bar{\lambda}_{1,1}=\bar{\lambda}$\;
    \For {$m=1:\log_2{N}$ } 
    {
    \For {$i=1:2^{m-1}$} 
    {
    $\bar{\lambda}_{2i-1,m+1}=\phi^{-1} (1-[1-\phi(\bar{\lambda}_{i,m})]^2)$\;
    $\bar{\lambda}_{2i,m+1}=2\bar{\lambda}_{i,m}$\;    
    }
    }
    \For {$\kappa=1:N$} 
    {
    $v_{\kappa}=Q\Big(\sqrt{\frac{\bar{\lambda}_{\kappa,\log_2(N)+1}}{2}}\Big)$
    }
           return $\textbf{v}$\;
    \caption{GA-BER($\bar{\lambda}$,$N$) (From \cite{Vangala2015}, modified)}
        \label{Algorithm2}
\end{function}

\begin{function}[!ht]
     \SetKwInOut{Input}{Input}
     \SetKwInOut{Output}{Output} 
     \Input{$\gamma$ [dB] and $N$}
    \Output{The maximum throughput code and the message length $K_{\text{SCD}}$}
 \DontPrintSemicolon   
   $\eta^{\prime}=1$, $\eta=0$, $\bar{N}_{0}^\kappa=0$, $P_{0}^{\kappa}=0$,$l=1$\;
   \While{$\eta^{\prime}\neq\eta$}
       { 
       $\bar{\lambda}=4l \times10^{\gamma/10}$\;
       $\textbf{v}$=GA-BER($\bar{\lambda}$,$N$)\Comment{Determine  BERs}\;
       \If{$l=1$}
       {
       [\textbf{v},\textbf{idx}]=Sort($\textbf{v}$) \Comment{Sort the bit-channels} \;
       }\Else{
        $\textbf{v}=\textbf{v}(\textbf{idx})$  \;
      }
       \For {$\kappa=1:N$}
       {
        $P_{l}^{\kappa}=1-\displaystyle  \prod_{i=1}^{\kappa}(1-v_i)$ \Comment{FER estimation}\;
        $\bar{N}_{l}^\kappa= \bar{N}_{l-1}^\kappa + \displaystyle N\prod_{l^{\prime}=1}^{l-1}P_{l^{\prime}}^{\kappa}$ \Comment{Effective code length} \;
        $\bar{\eta}_\kappa=\frac{\kappa}{\bar{N}_{l}^\kappa}(1-P_{l}^{\kappa}) $ \Comment{Throughputs}\;
       }
      $K_{\text{opt}}= \displaystyle \argmax_{\kappa} \bar{\eta}_\kappa$ \Comment{Find $K$ with  highest $\eta_\kappa$}\;
      $\eta^{\prime}=\eta$\;
     $\fontdimen16\textfont2=5pt\eta=\bar{\eta}_{K_{\text{opt}}}$ \Comment{Find  highest throughput} \fontdimen16\textfont2=3pt\;
      $l=l+1$ \Comment{Increment retransmission number}\;
      }
      $R_{\text{SCD}}=\bar{\eta}_{K_{\text{opt}}}$ \;
       return $\textbf{idx}$, $K_{\text{SCD}}$\;
       \caption{CC-Binary-Design($\gamma$,$N$)}
    \label{Algorithm33}
\end{function}

\begin{algorithm}[!ht]
 \DontPrintSemicolon
     \SetKwInOut{Input}{Input}
     \SetKwInOut{Output}{Output} 
     \Input{$\boldsymbol{\bar{\lambda}}_{\text{PAM}}$ for PAM, B, and N}
    \Output{The maximum throughput codes and $K_{\text{SCD}}$}
      \nonl \textbf{Procedures:} AverageLLR($\boldsymbol{\bar{\lambda}}_{\text{PAM}}$): Calculates $\bar{\lambda}$ of a QAM based on (\ref{PAMLLRn1}) for a PAM constellation.\;
       $\boldsymbol{\bar{\lambda}_{\text{QAM}}}$=AverageLLR($\boldsymbol{\bar{\lambda}}_{\text{PAM}}$)\;
       \For {$n=1:B$}
       {
       $\textbf{v}_n$=GA-BER($\bar{\lambda}_{\text{QAM},n}$,$N$) \;
       }
        [$\textbf{idx}_{\text{tot}}$,$K$]=NC-Binary-Design([$\textbf{v}_1,\textbf{v}_2,...,\textbf{v}_B]$,$BN$)\;
        return $\textbf{idx}_{\text{tot}}$, $K_{\text{SCD}}$ \;
    \caption{GA-based design for NC-D with QAM}
    \label{Algorithm5}
\end{algorithm}

\begin{algorithm}[!ht]
 \DontPrintSemicolon
     \SetKwInOut{Input}{Input}
     \SetKwInOut{Output}{Output} 
     \Input{$\boldsymbol{\bar{\lambda}}_{\text{PAM}}$: The average LLR vector for a PAM and $N$}
    \Output{The maximum throughput codes  
    and the message length for the SCD}
       $\boldsymbol{\bar{\lambda}}_\text{QAM}$=AverageLLR($\boldsymbol{\bar{\lambda}}_{\text{PAM}}$)\;
       \For {$n=1:B$} 
       {
        $\textbf{v}_n$=GA-BER($\bar{\lambda}_{\text{QAM},n}$,$N$)\;
       [$\textbf{idx}_n$,$K_{\text{SCD},n}$]=NC-Binary-Design($\textbf{v}_n$,$N$)\;
        }
       return $\textbf{idx}_1$,..., $\textbf{idx}_B$, $\bf{K}_{\text{SCD}}$\;
    \caption{GA-based design for NC-I with QAM}
    \label{Algorithm7}
\end{algorithm}

\begin{function}[!ht]
 \DontPrintSemicolon
    \SetKwInOut{Input}{Input}
   \SetKwInOut{Output}{Output} 
     \Input{$\gamma$, $K_{\text{SCD}}$, $N$, $\textbf{idx}$, and $B$}
    \Output{The maximum throughput and the corresponding code rate for the SCLD}
  \nonl \textbf{Procedures:} {f($K$)=SCLD\_Throughput($\gamma$,$K$,$N$,$\textbf{idx}$,$B$): Employs SCLD to simulate and estimate the throughput of a MLPCM with a total message length of $K$, a level-code length of $N$, the indices $\textbf{idx}$ designed for SCD, a $2^B$-QAM constellation and $\gamma$.}\;
 \nonl      Golden ratio: $\rho = \frac{\sqrt{5}-1}{2}$\;
   $a=K_{\text{SCD}}$ \;
   $b=\text{min}(a+\frac{BN}{10},BN)$ \;
    $k_1= \floor{\rho a+ (1-\rho)b}$ \Comment{First point to test} \;
     $f(k_1)$ ={SCLD\_Throughput($\gamma$,$k_1$,$N$,$\textbf{idx}$,$B$)}\;
     $k_2= \floor{(1-\rho)a + \rho b}$\Comment{Second point to test}\;
     $f(k_2)$ = {SCLD\_Throughput($\gamma$,$k_2$,$N$,$\textbf{idx}$,$B$)}\;
     \While { $\abs{a-b} > 1$ }{
   \eIf {$f(k_1) > f(k_2)$ }
       {
       $b=k_2$, $k_2=k_1$, $f(k_2)=f(k_1)$\;
       $k_1= \floor{\rho a+ (1-\rho)b}$ \;
       $f(k_1)$ = {SCLD\_Throughput($\gamma$,$k_1$,$N$,$\textbf{idx}$,$B$)}\;
    }
    {
       $a=k_1$, $k_1=k_2$, $f(k_1)=f(k_2)$\;
       $k_2= \floor{(1-\rho)a + \rho b}$ \;
       $f(k_2)$ = {SCLD\_Throughput($\gamma$,$k_2$,$N$,$\textbf{idx}$,$B$)}\;
    }
      }
      return $k_1$\;
    \caption{SCLD-Rate-Match($\gamma$,$K_{\text{SCD}}$,$N$,$\textbf{idx}$,$B$)}
 \label{Algorithm4}
\end{function}

\ifCLASSOPTIONcaptionsoff
  \newpage
\fi

\end{document}